\documentclass[11pt]{article}
\usepackage{amsmath}
\usepackage{amsfonts,dsfont}
\usepackage{amssymb}
\usepackage{amsthm}
\usepackage{mathtools}
\usepackage{thmtools}
\usepackage{bm}
\usepackage{verbatim}
\usepackage{xcolor}
\usepackage{physics}

\usepackage{authblk}
\usepackage{graphicx}
\usepackage[hidelinks]{hyperref}
\usepackage[english]{babel}
\usepackage{float}

\newlength{\bredde}
\def\slash#1{\settowidth{\bredde}{$#1$}\ifmmode\,\raisebox{.15ex}{/}
\hspace*{-\bredde} #1\else$\,\raisebox{.15ex}{/}\hspace*{-\bredde} #1$\fi}
\newcommand{\RE}{{\rm Re\,}}
\newcommand{\HE}{{\rm He\,}}
\newcommand{\Pf}{{\rm Pf\,}}

\newcommand{\IM}{{\rm Im\,}}

\newcommand{\eins}{\leavevmode\hbox{\small1\kern-3.8pt\normalsize1}}
\newcommand{\be}{\begin{equation}}
\newcommand{\ee}{\end{equation}}
\newcommand{\bee}{\begin{eqnarray}}
\newcommand{\eee}{\end{eqnarray}}

\newcommand{\eGinUE}{\textup{eGinUE}}
\newcommand{\eGinOE}{\textup{eGinOE}}

\newtheorem{thm}{Theorem}[section]

\newtheorem{cor}[thm]{Corollary}
\newtheorem{lem}[thm]{Lemma}
\newtheorem{prop}[thm]{Proposition}
\newtheorem*{prob*}{Problem}
\newtheorem*{thm*}{Theorem}

\theoremstyle{definition}
\newtheorem{defn}[thm]{Definition}

\newtheorem*{defn*}{Definition}
\newtheorem{rem}[thm]{Remark}

\newtheorem*{rem*}{Remark}
\numberwithin{equation}{section}
\newcommand{\erfc}{\text{erfc}}

\providecommand{\keywords}[1]
{
  \small	
  \textbf{Keywords:} #1
}

\title{Mean eigenvector self-overlap in the real and complex elliptic Ginibre ensembles at strong and weak non-Hermiticity}
\author{ M. J. Crumpton, Y. V. Fyodorov and T. R. W\"urfel } %
\affil{Department of Mathematics, King's College London, London WC2R 2LS, UK}

\date{\today}

\begin{document}

\maketitle
\begin{abstract}
\noindent
We study the mean diagonal overlap of left and right eigenvectors associated with complex eigenvalues in $N\times N$ non-Hermitian random Gaussian matrices. In well known works by Chalker and Mehlig the expectation of this (self-)overlap was computed for the complex Ginibre ensemble as $N\to \infty$. In the present work, we consider the same quantity in the real and complex elliptic Ginibre ensembles characterized by correlations between off-diagonal entries controlled by a parameter $\tau\in[0,1]$, with $\tau=1$ corresponding to the Hermitian limit. We derive exact expressions for the mean diagonal overlap in both ensembles at any finite $N$, for any eigenvalue off the real axis. We further investigate several scaling regimes as $N\rightarrow \infty$, both in the limit of strong non-Hermiticity keeping a fixed $\tau\in[0,1)$  and in the weak non-Hermiticity limit,  with $\tau$  approaching unity  in such a way that $N(1-\tau)$ remains finite.

\end{abstract}


\section{Introduction}\label{sec:intro}

Influence of non-Hermiticity in physical systems (both classical and quantum) has recently attracted significant attention, both theoretically and experimentally, see e.g. \cite{non-Herm_rev} for a review.  In many  situations the corresponding operators or matrices may have (or do have) a random component, either by intrinsic properties of the system or as a model reflecting lack (or irrelevance) of the knowledge of their exact structure. 
All these developments boosted attention to various classes of non-Hermitian random matrices, which  despite lasting interest in their properties \cite{KS}  remain much less understood in comparison with their self-adjoint counterparts. Motivations for such studies traditionally include, among other topics, attempts to characterize equilibria, stability and dynamics of complicated, mostly nonlinear dynamical systems, e.g. neural networks, or communities of many interacting species, see \cite{May72,SCSS,WT13,FK2016,MB2019,BAFK21}, as well as studies of an effective non-Hermiticity helping to characterize various aspects of scattering in systems with quantum or classical wave chaos nature \cite{FSom1,Rotter,FSav1}. More recently, a boost of interest in non-Hermitian random matrices arose due to attempts to understand better universal properties of  dissipative Lindbladian terms in density-matrix evolution of quantum many-body systems \cite{random_Lindblad,AKMP,SRP,LPC21}.

Without loss of generality, a random $N \times N$ matrix $X$ can be diagonalised as $X = S \Lambda S^{-1}$, where $\Lambda$ is a diagonal matrix containing the eigenvalues of $X$, represented by $z_i$. Additionally, the columns of $S$ are the right eigenvectors, $\bm v_{R,i}$ and the rows of $S^{-1}$ are the left-eigenvectors, $\bm v_{L,i}^\dagger$. Within this setup the left and right eigenvectors satisfy a bi-orthogonality condition $\bm v_{L,i}^\dagger \bm v_{R,j} = \delta_{ij}$,  but separately inside each of the sets of left and of right eigenvectors no such relation  exists, unless matrices are normal. Chalker and Mehlig \cite{CM,MC} were the first to systematically address the non-orthogonality of eigenvectors in non-normal random matrices, restricting their studies to the classical complex Ginibre ensemble (GinUE) \cite{Ginibre}. In particular, they introduced an important object, the \emph{matrix of overlaps}
\be
    \mathcal{O}_{mn} = \left( \bm v^\dagger_{L,m} \bm v_{L,n} \right) \left( \bm v^\dagger_{R,n} \bm v_{R,m} \right) \quad m, n=1,\ldots,N ,
    \label{eq:def_overlap}
\ee
and studied associated connected  ensemble averages for diagonal entries $ \mathcal{O}_{nn}$ (frequently called self-overlaps) and their off-diagonal counterparts  $\mathcal{O}_{mn}$:
\be\label{ChalkerMehligOverlap}
    \mathcal{O}(z) \equiv \bigg\langle \frac{1}{N}\sum_{n=1}^N \mathcal{O}_{nn} \ \delta(z-z_n) \bigg\rangle \hspace{0.75cm}  \text{and} \hspace{0.75cm} \mathcal{O}(z_1,z_2) \equiv \bigg\langle \frac{1}{N^2}\sum_{\substack{m,n=1 \\ m \neq n}}^N \mathcal{O}_{mn} \ \delta(z_1 -z_m) \  \delta(z_2-z_n) \bigg\rangle\ ,
\ee
where the angular brackets stand for the expectation value with respect to the measure associated with the ensemble in question, i.e. the GinUE in the case of \cite{CM}. Here $\delta (z-z_n)$ is the Dirac delta mass at the eigenvalue $z_n$, so that the empirical density of eigenvalues in the complex plane, $z$, reads $\rho^{(\text{emp})}_N(z) = \frac{1}{N} \sum_{n=1}^N \delta(z - z_n)$. 

The problem of statistical characterization of eigenvector non-orthogonality in various ensembles of non-Hermitian random matrices has been attracting growing research ever since the Chalker and Mehlig papers, in particular driven by many applications in physics and beyond.  More recently the interest in these questions got further impetus from papers by Bourgade \& Dubach \cite{BD} and Fyodorov \cite{FyodorovCMP}, who managed to derive, among other quantities, the explicit distribution of self-overlaps $\mathcal{O}_{nn}$ in the bulk of the complex Ginibre ensemble. Such self-overlaps are directly related to the eigenvalue condition numbers, controlling stability of eigenvalues against additive perturbation of the matrix entries, see e.g. the discussion in  \cite{FyodorovCMP}. Further representative examples of studies of eigenvector non-orthogonality in non-normal random matrices, of various types and at varying level of mathematical rigor, as well as some motivating applications, can be found in \cite{SS} - \cite{Tarnowski24}.

In this paper we will concentrate on the diagonal overlaps described via $\mathcal{O}(z)$.  This can be further used to define the mean conditional self-overlap as 
\be\label{conditional}
    \mathbb{E}\left( z  \right) \equiv \mathbb{E}\left(\mathcal{O}_{nn} \ \vert \ z = z_n \right) = \frac{\mathcal{O}(z)}{ \rho\, (z) } \ \text{,}
\ee
where $\rho\, (z)$ is the mean spectral density at a point $z$ in the complex plane, defined via
\begin{equation}\label{meanden}
    \rho\, (z) \equiv \langle \ \rho_N^{(\text{emp})}(z) \ \rangle = \bigg\langle \frac{1}{N}\sum_{n=1}^N  \delta(z-z_n) \bigg\rangle \ . 
\end{equation}
Note that $\mathbb{E}\left( z \right)$ is particularly useful when comparing theoretical predictions with numerical data. For the GinUE, as $N\rightarrow \infty$ and each entry having variance $1/N$, the droplet becomes a unit circle in $\mathbb{C}$ with uniform density $ \rho (z) \approx 1/\pi $ for $\vert z \vert^2 <  1 $ and zero otherwise \cite{Girko1}. Chalker and Mehlig found the leading asymptotic behaviour of Eq. \eqref{ChalkerMehligOverlap} in this ensemble, deriving $\mathcal{O}(z) \approx N \left(1-\vert z \vert^2 \right)/\pi$ inside the unit circle and zero otherwise \cite{CM}. Thus, $\mathcal{O}_{nn} \sim N$ for eigenvalues inside the disk, which is macroscopically larger than the corresponding value seen in normal matrices, where $\mathcal{O}_{nn} = 1$. One should note that in order to obtain a non-trivial limit as $N \to \infty$, we consider the well-defined rescaled correlator $ \mathcal{O}(z)/N $. 

In our recent paper \cite{WCF}, we provided explicit expressions for the mean self-overlap in the real Ginibre ensemble at finite matrix size $N$ and, for $N \rightarrow \infty$, we analyzed various regions associated with complex eigenvalues inside the limiting spectral support of the circular droplet. In the present work, we develop our analysis further by extending it to more general classes of random matrices with real or complex Gaussian-distributed entries. Namely, the so-called elliptic ensembles introduced by Girko \cite{Girko2}, who demonstrated that the limiting spectral support will have the shape of an elliptic droplet controlled by an ellipticity parameter $\tau \in [0,1]$. In particular, the length of the semi-major and semi-minor axes of the ellipse become proportional to $1+\tau$ and $1-\tau$ respectively.

Our main goal then is to study the mean self-overlap of eigenvectors associated with complex eigenvalues for the complex elliptic Ginibre ensemble (eGinUE) and the real elliptic Ginibre ensemble (eGinOE), both at finite $N$ and asymptotically as $N \rightarrow \infty$.  We will first consider the parameter $\tau$  fixed (and away from one) as $N \to \infty$, this defines the so-called regime of \emph{strong non-Hermiticity} (SNH). Note that for $\tau$ being zero we naturally recover the Ginibre ensembles, GinUE and GinOE respectively, and thus, in this limit, our results have to match those presented in \cite{WCF}. Additionally, we study another important scaling regime of \emph{weak non-Hermiticity} (WNH), originally introduced in \cite{FKS97a,FKS97b,FKS98}. In this regime the parameter $\tau \to 1$ as $N$ tends to $\infty$, more precisely $\tau \sim 1 - \kappa/N$, with $\kappa > 0$, which results in the elliptic droplet being replaced by a thin cloud of complex eigenvalues. In addition, for real-valued matrices a macroscopic fraction of eigenvalues remain exactly on the real line, while the complex eigenvalues are depleted away from the real axis. These features are absent in the complex-valued weakly non-Hermitian matrices, making results at weak non-Hermiticity pronouncedly different between real-valued and complex-valued matrices. 

The rest of the paper is organized as follows. After introducing the necessary notation and statements about random real and complex elliptic Ginibre matrices in Section \ref{subsec:DefRem}, we present our main results for the mean self-overlap in Section \ref{subsec:MainRes} for finite matrix size $N$ in the eGinUE (Theorem \ref{thm:MainResCOMPLEX}) and the eGinOE (Theorem \ref{thm:MainRes}). Asymptotic results for the mean self-overlap, as $N \rightarrow \infty$, are given at strong non-Hermiticity in Section \ref{subsubsec:StrongRegime} and at weak non-Hermiticity in Section \ref{subsubsec:WeakRegime}. We provide asymptotic formulas for the mean self-overlap in the eGinOE at SNH in the spectral bulk and in the special depletion regime, which only exists in the eGinOE, in Corollaries \ref{cor:eGinOEbulkStrong} and \ref{cor:eGinOEdepletionStrong} respectively. Furthermore, we discuss the mean self-overlap at WNH in what we call the \emph{weak bulk} for the eGinUE in Corollary \ref{cor:eGinUEbulkWeak} and the eGinOE in Corollary \ref{cor:eGinOEbulkWeak}. Finally, the corresponding proofs of our findings are presented in Section \ref{sec:MainResDerivation} for finite $N$ and in Section \ref{sec:AsymptoticAnalysis} for the asymptotic results.


\section{Statement and Discussion of the Main Results}\label{sec:MainRes} 



\subsection{Remarks on elliptic Ginibre ensembles}
\label{subsec:DefRem}

\begin{defn}\label{def:eGinUE}
    Let $X = \left( X_{ij} \right)_{i,j=1}^N$ be an $N \times N$ matrix in the \emph{complex elliptic Ginibre} ensemble denoted as eGinUE. Matrices in this ensemble consist of mean zero i.i.d. complex Gaussian entries \cite{Girko2,SCSS}, with the following correlation structure
    \be\label{eq:corr_eGinUE}
        \mathbb{E}\left( X_{ii}^2  \right) = \mathbb{E}\left( X_{ij}^2  \right)= 1
        \hspace{1cm} \mathbb{E}\left( X_{ij} X_{ji} \right) = \tau \text{,} \quad \text{for} \quad i,j=1,...N \quad i\neq j   \ ,
    \ee
    for $\tau \in [0,1]$. The joint probability density function (JPDF) of such matrices, with respect to the flat Lebesgue measure, $dX = \prod_{i,j=1}^{N} dX_{ij}d\bar{X}_{ij}$, is given by
    \be\label{jpdfReGinCOMPLEX}
        P_{\text{eGinUE}}(X)\,dX = \frac{1}{D_{N,\tau}} \exp \left[ -\frac{1}{1-\tau^2} \Tr \left( XX^\dagger -\tau \ \RE X^2 \right)  \right]dX \ ,
    \ee
    with normalization constant, $D_{N,\tau}= \pi^{N^2} \left( 1-\tau^2 \right)^{\frac{N^2}{2}}$ \cite{FKS97b}. Matrices belonging to this ensemble can be generated according to $X = \sqrt{1 + \tau} \, H_1 + i \sqrt{1 - \tau} \, H_2$, where $H_1$ and $H_2$ are complex Hermitian matrices with mean zero Gaussian entries. The real entries on the diagonal have variance $1/2$ whereas the complex off-diagonal entries have variance $1/4$ for both their real and imaginary components.
\end{defn}

\begin{defn}\label{def:eGinOE}
    The counterpart of the eGinUE for matrices with real i.i.d. Gaussian entries is the so-called \emph{real elliptic Ginibre} ensemble, denoted as eGinOE. Such matrices have entries of mean zero, with the following correlation structure
    \be\label{eq:corr_eGinOUE}
        \mathbb{E}\left( X_{ii}^2  \right) = 1+ \tau \hspace{1cm} \mathbb{E}\left( X_{ij}^2  \right)= 1
        \hspace{1cm} \mathbb{E}\left( X_{ij} X_{ji} \right) = \tau \text{,} \quad \text{for} \quad i,j=1,...N \quad i\neq j   \ ,
    \ee
    for $\tau \in [0,1]$. The JPDF of matrices $X$ with respect to the flat Lebesgue measure \cite{Girko2,SCSS}, reads
    \be\label{jpdfReGin}
        P_{\text{eGinOE}}(X)\,dX = \frac{1}{C_{N,\tau}} \exp \left[ -\frac{1}{2(1-\tau^2)} \Tr \left( XX^T -\tau X^2 \right)  \right] dX \ ,
    \ee
    with normalization constant $C_{N,\tau}= (2\pi)^{N^2/2} \left( 1+\tau \right)^{\frac{N(N+1)}{4}} \left( 1- \tau \right)^{\frac{N(N-1)}{4}}$ \cite{FT}. In a similar way to the eGinUE, matrices belonging to this ensemble can be randomly sampled using $X = \sqrt{1 + \tau} \, H + \sqrt{1 - \tau} \, A$, where $H$ and $A$ are real symmetric and anti-symmetric Gaussian matrices respectively. Both have mean zero entries and variance $1/2$ of their constrained off-diagonal entries, however the diagonal entries of $H$ have unit variance whereas, by definition, $A$ must have all zeros on the diagonal.  
\end{defn}

\noindent
For recent reviews of available results on the eGinUE and eGinOE see \cite{BF1,BF2}. The remarks below provide a discussion of main facts on both ensembles which will be of direct relevance to the present paper.

\begin{rem}\label{rem:GinOEDef}
    The correlation parameter $\tau\in [0,1]$ allows one to interpolate between the two classical limits. For $\tau = 0$, the entries of eGinOE and eGinUE matrices are uncorrelated, as such the eGinOE and eGinUE reduce to the GinOE \cite{Ginibre, BF2} and the GinUE \cite{Ginibre, BF1} respectively. For $\tau = 1$, the correlation is at its maximal possible value, yielding symmetric matrices belonging to the \textit{Gaussian Orthogonal Ensemble} (GOE) in the real case, whereas in the complex case one arrives at Hermitian matrices from the \textit{Gaussian Unitary Ensemble} (GUE). Finite $N$ and asymptotic results for the mean density of complex eigenvalues and the mean diagonal overlaps for both the GinOE and GinUE can be found in \cite{WCF}.
    Since the characteristic polynomial of an eGinOE matrix has only real coefficients, generically the spectrum of eGinOE matrices consists of real eigenvalues $\lambda\in \mathbb{R}$ and complex eigenvalues $z=x+iy, \, y\ne 0$ which always appear in conjugated pairs. Correspondingly, the mean spectral density, Eq. \eqref{meanden}, necessarily has the form $\rho\,(z) = \rho^{(c)}(z)+\delta(y) \rho^{(r)}(x)$, where $\rho^{(c/r)}(z)$ describes the mean density of complex/real eigenvalues. On the other hand, the eigenvalues of eGinUE matrices are all complex and do not respect complex conjugation with probability one.
\end{rem}

\noindent
The density of eigenvalues in both ensembles has already been studied extensively. Before presenting the explicit formulas, we need to introduce the Hermite polynomials via:
\be\label{Eq:HermitePoly}
    \HE_k(x) = \frac{1}{\sqrt{2\pi}} \int_{-\infty}^{\infty} dy \, e^{-\frac{1}{2}y^2} \, \left( x+iy \right)^k =  \frac{(\pm i)^k}{\sqrt{2\pi}} e^{\frac{x^2}{2}}\int_{-\infty}^{\infty} dt \,  e^{-\frac{1}{2}t^2 \mp ixt } t^k = k! \sum_{m=0}^{[k/2]} \frac{(-1)^m}{(k - 2m)! m!} \frac{x^{k - 2m}}{2^m}\ ,
\ee
where $[x]$ represents the floor function. We will also need the standard complementary error function $\erfc(x) = 1 - \erf(x)$, where $\erf(x) = \frac{2}{\sqrt{\pi}} \int_0^x e^{-t^2} dt$. 
    
\begin{rem}\label{rem:densityExpressions}
    The mean density of complex eigenvalues is of particular interest in this work and is given for finite $N$ 
    in the eGinUE by
    \be\label{eq:eGinUE_density}
        \rho^{(\text{eGinUE},c)}_N(z) = \frac{1}{\pi} \frac{1}{\sqrt{1-\tau^2}}\exp \left[ -\frac{1}{1-\tau^2}\bigg( \vert z \vert^2 - \tau \ \text{Re}(z^2) \bigg) \right] \sum_{n=0}^{N-1} \frac{\tau^n}{n!}\HE_n\left( \frac{z}{\sqrt{\tau}} \right)\HE_n\left( \frac{\bar{z}}{\sqrt{\tau}} \right) \ ,
    \ee
    equivalent representations for finite $N$ are also known, see e.g. \cite[Prop. 2.7]{BF1} and \cite[Eq. (18.4.19)]{KS}.  On the other hand, in the eGinOE the corresponding expression reads
    \be\label{eGinOEdensityComplex}
        \rho_N^{(\text{eGinOE,c})}(z)= \sqrt{\frac{2}{\pi}}\frac{1}{1+\tau} \ y \ \exp \left[ \frac{y^2-x^2}{1+\tau} \right] \text{erfc}\left[ \sqrt{\frac{2}{1-\tau^2}} \ \vert y \vert \right] P_{N-2} \ ,
    \ee
    where we have introduced
    \begin{equation}
         P_N =  \frac{1}{\bar{z}-z} \sum_{k=0}^{N} \ \frac{\tau^{k+\frac{1}{2}}}{k!} \bigg[ \HE_{k+1}\left(\frac{\bar{z}}{\sqrt{\tau}}\right) \HE_{k}\left( \frac{z}{\sqrt{\tau}} \right) - \HE_{k+1}\left( \frac{z}{\sqrt{\tau}} \right) \HE_{k}\left(\frac{\bar{z}}{\sqrt{\tau}}\right) \bigg]\text{,}
         \label{Eq:PNRes}
    \end{equation}    
    see e.g. \cite[Eq. (5.2)]{FN2} and \cite[Eq. (23)]{APS} for equivalent representations of this density. The mean density of \emph{real} eigenvalues in the eGinOE is also known at finite $N$, see e.g. \cite{FN2,ForM} or \cite[Remark 2.3]{FT}, but is not needed for our purposes.
\end{rem}

\noindent
It is worth discussing the large $N$ asymptotic behaviour of the mean eigenvalue density in more detail. Focusing first on the regime of \emph{strong non-Hermiticity}, defined such that $0\leq \tau <1$ remains fixed as $N \to \infty$, we show a comparison between the large $N$ densities of complex eigenvalues in the eGinOE and eGinUE in the upper half of Figure \ref{fig:Heatmaps} using heatmaps. The heatmap associated with the eGinUE features two distinct scaling regimes: a spectral bulk inside the ellipse and an edge region along the ellipse circumference. Similar regimes are also seen in the eGinOE. In both ensembles it has previously been shown \cite{Girko2,SCSS,FN2}, that the limiting mean density of complex eigenvalues in the spectral bulk, after rescaling  $z=\sqrt{N}w$ with $w=x+iy$, converges to 
\be\label{eq:density_bulk_strong}
    \rho_{\text{bulk,SNH}}^{(\eGinUE, c)}(w) =\rho_{\text{bulk,SNH}}^{(\eGinOE,c)}(w) = \frac{1}{\pi(1-\tau^2)} \ \Theta\bigg[1-\left(\frac{x^2}{(1+\tau)^2} +\frac{y^2}{(1-\tau)^2}\right) \bigg]   \ ,
\ee
where $\Theta[a]=1$ when $a > 0$ and zero otherwise is the Heaviside step function. This implies that in the limit of large $N$, for both the eGinUE and eGinOE, the vast majority of eigenvalues lie inside the ellipse
\begin{equation}
    \frac{\RE(z)^2}{N(1 + \tau)^2} + \frac{\IM(z)^2}{N(1 - \tau)^2} = 1 \ ,
    \label{eq:ellipse}
\end{equation} 
thus, at SNH, both the semi-major and semi-minor axes of the elliptical support are $O(\sqrt{N})$. 

\begin{figure}[h!]
    \centering
    \includegraphics[scale = 0.335]{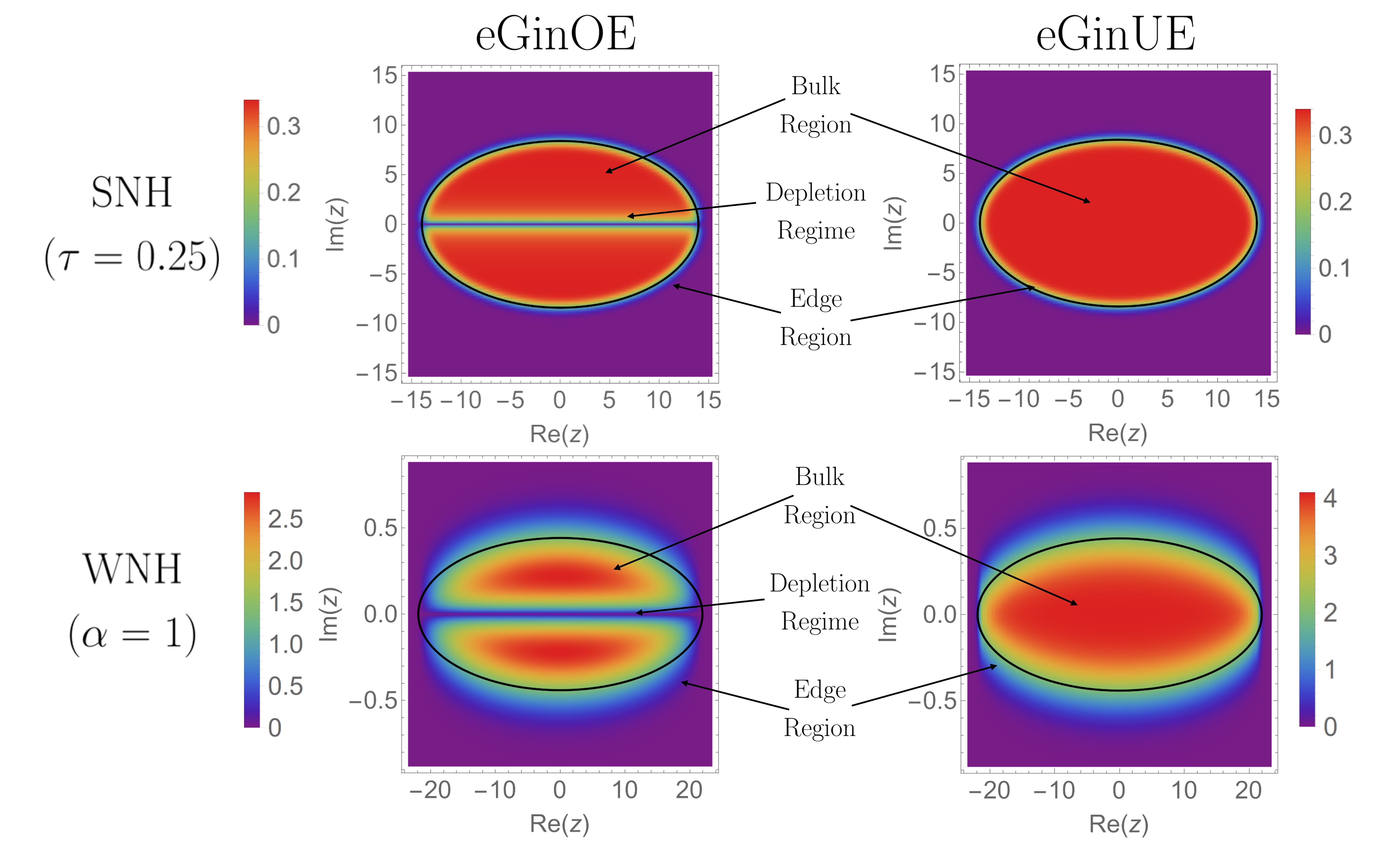}
    
    \caption{\small Large $N$ density of complex eigenvalues in the eGinOE (left) and eGinUE (right) in the strong (top) and (bottom) weak non-Hermiticity regimes with relevant scaling regions labelled. Strong non-Hermiticity plots are shown for $\tau = 0.25$ and weak non-Hermiticity plots are for $\alpha = 1$. In all plots the solid black ellipse is given by Eq. \eqref{eq:ellipse} and the density is depicted on a rainbow scale. This diagram explains schematically what is meant by the bulk, edge and depletion regimes of large eGinOE and eGinUE matrices. Each heatmap is plotted using the finite $N$ expression for the density of complex eigenvalues in the associated ensembles when $N=125$. }
    \label{fig:Heatmaps}
\end{figure}

\noindent
In addition, the heatmap in the eGinOE case shows yet another scaling regime close to the real axis, essentially for $\IM(z) \sim O(1)$. In that region the density of eigenvalues is significantly reduced in comparison to the spectral bulk value, reflecting the presence of the $O(\sqrt{N})$ real eigenvalues in the eGinOE, see \cite{WCF} for discussion in the context of the GinOE. The scaling regime describing this region of the complex plane will be referred to as the \emph{depletion regime}. Inside the depletion regime of the eGinOE, where we write $z=\sqrt{N}\delta + i\xi$, with $\delta,\xi \sim O(1)$, the mean density of complex eigenvalues in the limit $N\to \infty$ converges to    
\be\label{Eq:density_depletion_strip_strong}
    \rho_{\text{depletion,SNH}}^{(\eGinOE,c)}(\delta,\xi) = \sqrt{\frac{2}{\pi}} \, \frac{1}{(1-\tau^2)^{\frac{3}{2}}}  \ \xi \ \exp \left[ \frac{2\, \xi^2}{1-\tau^2} \right] \, \textup{erfc}\left[ \sqrt{\frac{2}{1-\tau^2}} \, \vert \xi \vert \right]\ \Theta \left[ 1 - \frac{\delta^2}{(1+\tau)^2} \right] \ .
\ee
The details of the derivation are given for convenience of the reader at the end of Section \ref{subsubsec:ProofAsymptdepletion}. The magnitude of a non-zero density at fixed $z$ in this regime does not depend on its real part, but only on $\xi$. The dependence on $\delta$ via the Heaviside step-function ensures that we stay inside the elliptic droplet, while the fact that $\xi$ remains unscaled with respect to $N$ restricts us inside a strip of width $O(1)$. Note that for $\tau = 0$ in Eqs. \eqref{eq:eGinUE_density} to \eqref{Eq:density_depletion_strip_strong} we obtain the corresponding expressions for the density in the GinOE and eGinUE \cite{WCF}. 

We now focus on the second scaling of $\tau$, the \emph{weak non-Hermiticity} regime. In contrast to strong non-Hermiticity, here $\tau$ is scaled with $N$ and approaches one as $N \rightarrow \infty$, specifically $\tau = 1-(\pi\alpha)^2/(2N)$. The density of complex eigenvalues at WNH is plotted using heatmaps in the lower half of Figure \ref{fig:Heatmaps} for both the eGinOE and eGinUE. From these plots one can see that, the elliptical support collapses to the real line in both ensembles, as they asymptotically approach the GOE and GUE respectively. In this region of weak non-Hermiticity, one can also see that both the eGinUE and eGinOE possess a spectral bulk and edge region and that the eGinOE has an additional region of eigenvalue depletion. It is also apparent from the lower half of Figure \ref{fig:Heatmaps} that there is an overall reduced density of complex eigenvalues in the eGinOE when compared to the eGinUE. This is due to the presence of a large number of real eigenvalues in the eGinOE. 

In the weak non-Hermiticity regime of the eGinOE and eGinUE, we can describe the behaviour of the vast majority of eigenvalues using the scaling $z = \sqrt{N}X + (i \pi y)/\sqrt{N}$, with $y \neq 0$ in the case of the eGinOE. For the eGinUE the spectral density in this regime and scaling reads
\be\label{Eq:eGinUEdensityweakbulk}
\begin{split}
    \rho_{\text{bulk,WNH}}^{(\eGinUE, c)}(X, y) =& \lim_{N\rightarrow \infty} \frac{\pi^2}{N} \rho^{(\text{eGinUE},c)}_N\left(z =  \sqrt{N} X + \frac{i \pi y}{\sqrt{N}} \right)
    = \frac{\sqrt{2}}{\pi^{3/2}} \, \frac{1}{|\alpha|} \, e^{-\frac{2y^2}{\alpha^2}} \,  \, \int_{0}^{\pi \sqrt{1 - \frac{X^2}{4}}} du \ e^{-\frac{\alpha^2 u^2}{2}} \cosh(2yu) \ .
\end{split}
\ee
This expression will simply follow from the derivation of the mean self-overlap in such a regime, see Eq. \eqref{Eq:eGinUEbulkweakproof} and is equivalent to \cite[Theorem 3]{ACV} after straightforward manipulation and rescaling $\alpha \to \alpha/\pi \nu(X)$ and $y \to \pi \nu(X)y$, where $\nu(X) = \sqrt{4 - X^2}/(2\pi)$. Additionally, one recovers \cite[Eq. (2.40)]{BF1} and \cite[Eq. (25)]{FKS98} in the vicinity of $X=0$. The corresponding expression in the eGinOE reads
\be\label{Eq:eGinOEdensityweakbulk}
\begin{split}
    \rho_{\text{bulk,WNH}}^{(\eGinOE,c)}(X,y) =&  \lim_{N\rightarrow \infty}  \frac{\pi^2}{N} \rho^{(\text{eGinOE},c)}_N\left(z = \sqrt{N}X + \frac{i \pi y}{\sqrt{N}} \right) 
    = \frac{1}{\pi} \text{erfc}\left[ \frac{\sqrt{2} |y|}{|\alpha|} \right] \, \int_{0}^{\pi\sqrt{1 - \frac{X^2}{4}}} du \, u \, e^{-\frac{\alpha^2 u^2}{2}} \, \sinh{\left(2 y u \right)} 
\end{split}
\ee
and is equivalent to results presented in \cite{FN2,Efe97} in the vicinity of $X=0$.  This result will be recovered as a byproduct of our derivation of the mean self-overlap  outlined in Section \ref{subsubsec:ProofeGinOEbulkWeak}. In Figure \ref{fig:density_WNH}, we plot the normalised conditional density of complex eigenvalues at WNH, 
\begin{equation}\label{condden}
    \tilde{\rho}_{\text{WNH}}\Big( X | y = \tilde{y} \Big) = \frac{\rho^{(c)}_{\text{bulk,WNH}}(X,\tilde{y})}{\int_{-2}^2 \rho^{(c)}_{\text{bulk,WNH}}(X,\tilde{y}) dX}  \hspace{0.5cm} \text{and} \hspace{0.5cm}\tilde{\rho}_{\text{WNH}}\Big( y | X = \tilde{X} \Big) = \frac{\rho^{(c)}_{\text{bulk,WNH}}(\tilde{X},y)}{\int_{-\infty}^\infty \rho^{(c)}_{\text{bulk,WNH}}(\tilde{X}, y) dy} \ ,
\end{equation}
in both the eGinUE and eGinOE at fixed $X$ and $y$. From these plots one can see that in both ensembles the single equation for the density of complex eigenvalues describes all interesting regimes apart from the edge. 

\begin{figure}[h!]
    \centering
    \includegraphics[scale=0.29]{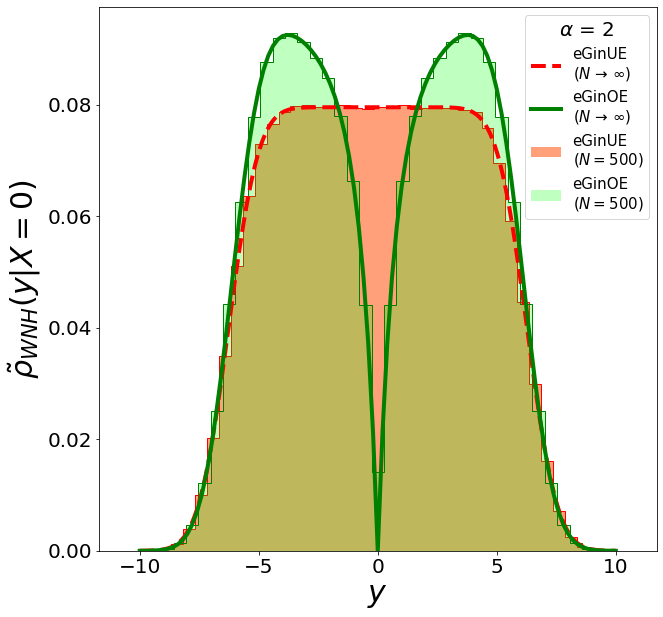}
    \hspace{1cm}
    \includegraphics[scale=0.29]{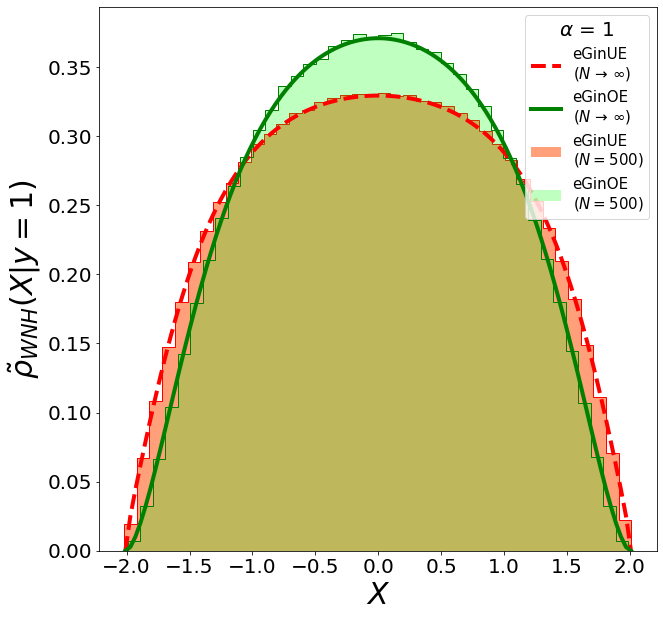}
    \caption{\small Normalised conditional density of complex eigenvalues, $\Tilde{\rho}_{\text{WNH}}$, see Eq.(\ref{condden}), at WNH in the eGinOE (green) and eGinUE (red). Left: $\Tilde{\rho}_{\text{WNH}}$ as a function of y, at fixed $X=0$ and $\alpha = 2$. Right: $\Tilde{\rho}_{\text{WNH}}$ as a function of X, at fixed $y=1$ and $\alpha = 1$. Each plot contains a normalised histogram of eigenvalues within $\pm 1/\sqrt{N}$ of the fixed $X$ or $y$, selected from an original set containing $O(10^8)$ eigenvalues.}
    \label{fig:density_WNH}
\end{figure}

\noindent
After the digressions on the mean densities, let us return to our main object of interest, the self-overlap between left and right eigenvectors. Results for the mean self-overlap are well established in the GinUE for the entire complex plane, see \cite{CM,MC,WS,BD,FyodorovCMP}. In the GinOE and eGinOE, results for real eigenvalues, i.e. when $z=x$, were discussed in \cite{FyodorovCMP, FT} and in the recent work \cite{WCF} some results for complex eigenvalues in the GinOE were found. For complex eigenvalues however, neither the eGinUE nor the eGinOE have been studied to the best of our knowledge. Hence, we now provide results for the mean self-overlap of eigenvectors associated with complex eigenvalues at finite $N$ and in several large $N$ regions of the droplet for the eGinOE and eGinUE.


\subsection{Statement of Main Results}
\label{subsec:MainRes}

We present our main findings in this section, giving the mean self-overlap of eigenvectors associated with a complex eigenvalue at finite $N$ for the eGinUE - see Theorem \ref{thm:MainResCOMPLEX} - and for the eGinOE - see Theorem \ref{thm:MainRes}. We also present results of asymptotic analysis in various scaling regimes. Proofs and technical details of the finite $N$ results and asymptotic analysis are provided in Sections \ref{sec:MainResDerivation} and \ref{sec:AsymptoticAnalysis} respectively.

\begin{thm}\label{thm:MainResCOMPLEX}
    Let $X$ be an $N \times N$ random matrix, with $N\geq 2$, drawn from the eGinUE. The mean self-overlap of left and right eigenvectors, Eq. \eqref{ChalkerMehligOverlap}, associated with a complex eigenvalue $z$ at finite $N$ is given by
    \be\label{Eq:MainResCOMPLEX}
    \begin{split}
        \mathcal{O}^{(\eGinUE,c)}_{N}(z) &= \bigg\langle \frac{1}{N}\sum_{n=1}^N \mathcal{O}_{nn} \ \delta(z-z_n) \bigg\rangle_{\textup{eGinUE},N}  \\
        &= \rho^{(\textup{eGinUE},c)}_{N}(z) + (1-\tau^2) \bigg[ \,  \rho^{(\textup{eGinUE},c)}_{N-1}(z) + (N-2) \, \rho^{(\textup{eGinUE},c)}_{N-2}(z) - R_{N-3} \, \bigg] \ ,
    \end{split}
    \ee
    where $\rho^{(\textup{eGinUE},c)}_{N}(z)$ is given in Eq. \eqref{eq:eGinUE_density} and the function $R_N\equiv R_N(z;\tau)$, is defined in terms of the Hermite polynomials, Eq. \eqref{Eq:HermitePoly}, as
    \be\label{Eq:RNRes}
    \begin{split}
        R_N &
        \equiv \frac{1}{\pi } \ \frac{1}{\sqrt{1-\tau^2}}  \ \exp \left[ -\frac{1}{1-\tau^2} \left( \vert z \vert^2 - \tau \, \RE (z^2) \right) \right] \sum_{k=0}^{N} \ \frac{k \ \tau^{k}}{k!} \, \HE_k\left(\frac{\bar{z}}{\sqrt{\tau}}\right) \HE_{k}\left( \frac{z}{\sqrt{\tau}} \right)   \ .
    \end{split}
    \ee
\end{thm}

\begin{thm}\label{thm:MainRes}
    Let $X$ be an $N \times N$ random matrix, with $N\geq 2$, drawn from the eGinOE. The mean self-overlap of left and right eigenvectors, Eq. \eqref{ChalkerMehligOverlap}, associated with a complex eigenvalue $z=x+iy$ at finite $N$ reads
    \be\label{Eq:MainRes}
    \begin{split}
        \mathcal{O}^{(\eGinOE,c)}_{N}(z) &= \bigg\langle \frac{1}{N}\sum_{n=1}^N \mathcal{O}_{nn} \ \delta(z-z_n) \bigg\rangle_{\textup{eGinOE},N}  = \frac{1}{\pi} \ \sqrt{\frac{1-\tau}{1+\tau}} \ \exp \left[ -\frac{1}{1+\tau}x^2 \right] \ \exp \left[-\frac{1}{1-\tau}y^2 \right] \\
        &\times \bigg[ 1 + \sqrt{\frac{\pi(1-\tau^2)}{2}} \ \exp \left[ \frac{2 y^2}{1-\tau^2} \right] \, \frac{1}{2\vert y \vert} \ \textup{erfc}\left( \sqrt{\frac{2 }{1-\tau^2}} \ \vert y \vert \right) \bigg] \\
        &\times \bigg[   P_{N-2} + (1-\tau^2)\bigg( P_{N-3} + (N-3) P_{N-4} -  T_{N-4} \bigg) \bigg] \ ,
    \end{split}
    \ee
    where the function $P_N\equiv P_N(z;\tau)$ is defined in Eq. \eqref{Eq:PNRes} and $T_N\equiv T_N(z,\tau)$ is defined in terms of the Hermite polynomials, Eq. \eqref{Eq:HermitePoly}, as 
    \be\label{Eq:TNRes}
    \begin{split}
        T_N &\equiv \frac{1}{\bar{z}-z} \sum_{k=0}^{N} \ \frac{k \ \tau^{k+\frac{1}{2}}}{k!} \bigg[ \HE_{k+1}\left(\frac{\bar{z}}{\sqrt{\tau}}\right) \HE_{k}\left( \frac{z}{\sqrt{\tau}} \right) - \HE_{k+1}\left( \frac{z}{\sqrt{\tau}} \right) \HE_{k}\left(\frac{\bar{z}}{\sqrt{\tau}}\right) \bigg] \ .
    \end{split}
    \ee
    
\end{thm}

\noindent
In Figure \ref{fig:finite_N}, we compare the mean conditional self-overlap in the eGinOE and eGinUE at finite $N$ given by Eqs. \eqref{Eq:MainResCOMPLEX} and \eqref{Eq:MainRes} with results of direct numerical simulations. We now discuss results in the limit $N\to \infty$, firstly in the bulk and depletion regime for the strong non-Hermiticity regime in Section \ref{subsubsec:StrongRegime} (eGinOE only), and then in the weak non-Hermiticity regime in Section \ref{subsubsec:WeakRegime} for both the eGinOE and eGinUE.

\begin{figure}[h]
    
    \centering
    \includegraphics[scale = 0.3]{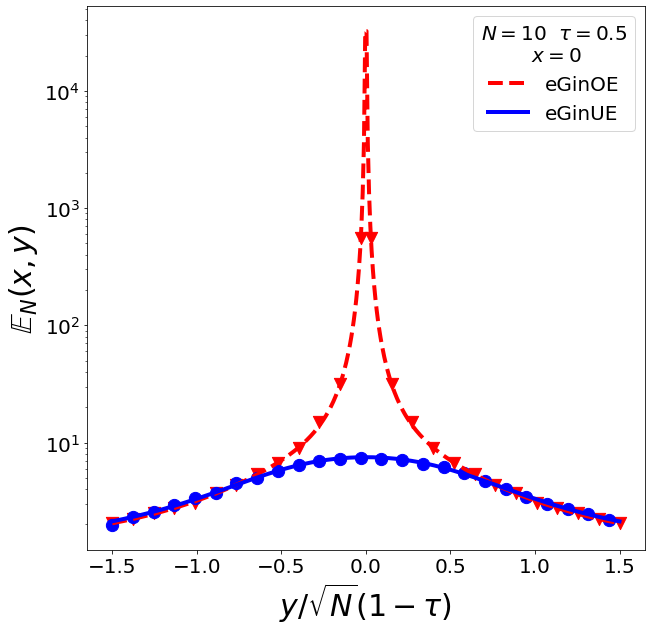}
    \hspace{2cm}
    \includegraphics[scale = 0.3]{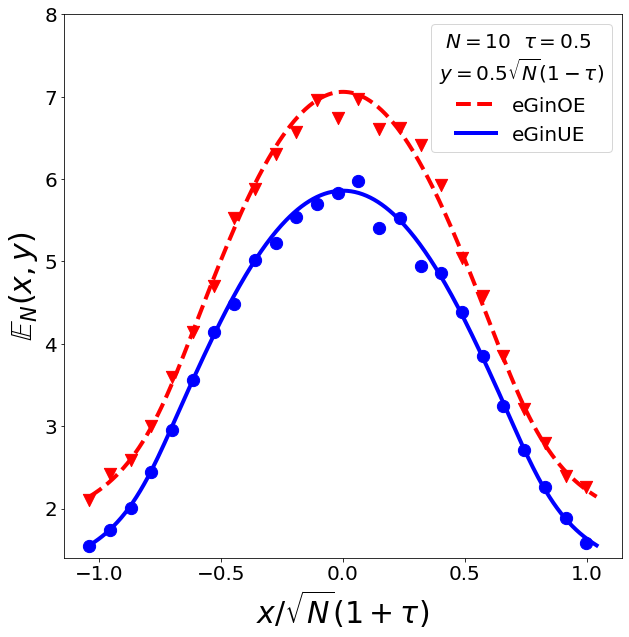}

    \caption{\small Mean conditional self-overlap, $\mathbb{E}_N(x,y)$, see Eq.(\ref{conditional}), of eigenvectors associated with complex eigenvalues in the eGinOE and eGinUE for $N = 10$ at fixed $\tau = 0.5$. Triangular (circular) markers represent numerical results for the eGinOE (eGinUE). Left: $\mathbb{E}_N(x,y)$ associated with purely imaginary eigenvalues as a function of $y$ at $x=0$. Right: $\mathbb{E}_N(x,y)$ associated with complex eigenvalues as a function of $x$, at a fixed $y = 0.5\sqrt{N}(1 - \tau)$. Numerically, the mean self-overlap is measured using the self-overlaps associated with the $O(10^3)$ eigenvalues nearest to the appropriate $x$ and $y$ from a data set containing $O(10^9)$ samples.}
    \label{fig:finite_N}
\end{figure}


\subsubsection{Strong non-Hermiticity Regime}
\label{subsubsec:StrongRegime}

In the strong non-Hermiticity regime the parameter $\tau$  remains fixed as the matrix size $N \to \infty$. One then easily arrives at the following
\begin{cor}\label{cor:eGinOEbulkStrong}
    For a complex eigenvalue $z = \sqrt{N}w$, where $w=x+iy$ ($y\neq 0$ and $|y| > N^{-1/2}$) and $\vert w \vert < 1$, the limiting scaled mean self-overlap of left and right eigenvectors in the bulk of the eGinOE, for a fixed $0\leq \tau < 1$, reads
    \be\label{Eq:eGinOEbulkStrongRes}
    \begin{split}
        \mathcal{O}_{\textup{bulk,SNH}}^{\textup{(eGinOE,c)}}(w) &\equiv \lim_{N\rightarrow \infty} \frac{1}{N} \ \mathcal{O}_N^{\textup{(eGinOE,c)}}\left( \sqrt{N}w \right) \\
        &= \frac{1}{\pi} \bigg[ 1 - \bigg( \frac{ x^2}{(1+\tau)^2} + \frac{ y^2}{(1-\tau)^2} \bigg)  \bigg] \ \Theta \left[1- \left( \frac{x^2}{(1+\tau)^2} + \frac{y^2}{(1-\tau)^2} \right)\right] \ .  
    \end{split}
    \ee
\end{cor}

\begin{rem}
The above result in the bulk of the eGinOE coincides with the same result in the bulk of the eGinUE \cite{JNNPZ,MC}, thus confirming the expected universality in this limit. In Figure \ref{fig:bulk_SNH}, we plot the mean conditional self-overlap of eigenvectors in the bulk, $\mathbb{E}_{\text{bulk}}(x,y)$, of the eGinOE and eGinUE in two different ways and compare to numerical simulations. Firstly, we plot $\mathbb{E}_{\text{bulk}}(x,y)$ as a function of $y$ along a fixed $x=0$, here we see a strong agreement for the eGinUE at all $y$. However, unsurprisingly the bulk result for the eGinOE breaks down as $y \to 0$, indicating the depletion effects not taken into account properly. Secondly, we plot $\mathbb{E}_{\text{bulk}}(x,y)$ as a function of $x$ along a fixed $y = 0.5(1 - \tau)$ (i.e. well inside the elliptic support), in this case we see a strong agreement between theory and simulations in both ensembles for all values of $x$. 
\end{rem}

\begin{figure}[h]
    \centering
    \includegraphics[scale = 0.3]{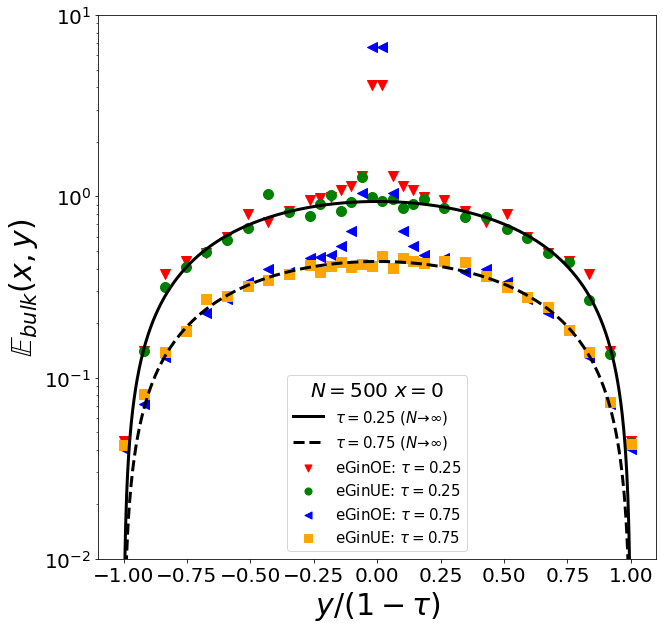}
    \hspace{2cm}
    \includegraphics[scale = 0.3]{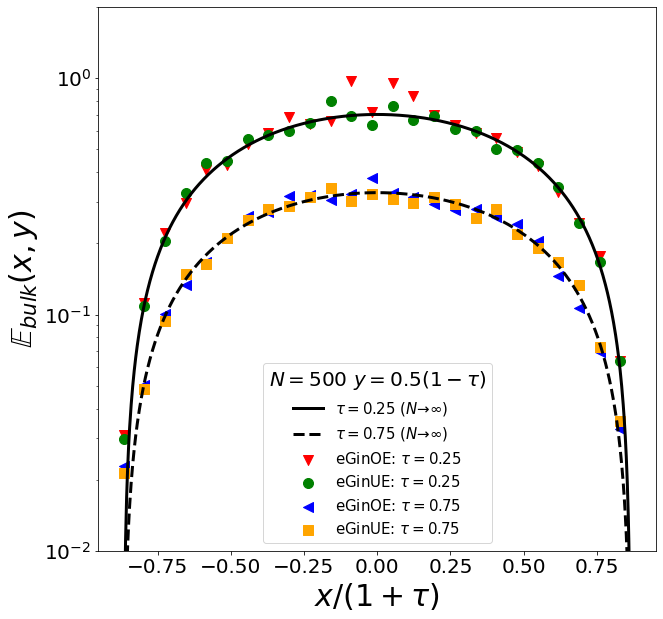}
    
    \caption{\small Mean conditional self-overlap of eigenvectors within the bulk, $\mathbb{E}_{\text{bulk}}(x,y)$, for the eGinUE and eGinOE at large $N$ for $\tau = 0.25$ (solid lines) and $\tau = 0.75$ (dashed lines). Left: $\mathbb{E}_{\text{bulk}}(x,y)$ at fixed $x=0$ plotted as a function of $y$. Right: $\mathbb{E}_{\text{bulk}}(x,y)$ at fixed $y=0.5(1 - \tau)$ and plotted as a function of $x$. In each of these plots we have compared our finite-N formulas for $\mathbb{E}_{\text{bulk}}(x,y)$ to numerical simulations (coloured markers) of eGinOE and eGinUE matrices of size $N=500$, obtained by averaging the $O(10^3)$ self-overlaps associated with complex eigenvalues closest to the chosen $x$ and $y$ from a set of $O(10^8)$ samples.}
    \label{fig:bulk_SNH}
\end{figure}

\begin{cor}\label{cor:eGinOEdepletionStrong}
    For a complex eigenvalue $z=\sqrt{N}\delta + i\xi$, such that $\delta, \xi \sim O(1)$ the limiting scaled mean self-overlap of left and right eigenvectors in the depletion region of the eGinOE, for a fixed $0\leq \tau < 1$, reads
    \begin{align}
        \mathcal{O}^{\textup{(eGinOE,c)}}_{\textup{depletion,SNH}}(\delta,\xi) &\equiv \lim_{N\rightarrow \infty} \frac{1}{N} \ \mathcal{O}_N^{\textup{(eGinOE,c)}}\left( \sqrt{N} \delta +i\xi \right) \label{Eq:eGinOEdepletionStrongRes} \\
        \nonumber&= \frac{1}{\pi} \Bigg[ 1 + \sqrt{\frac{\pi(1-\tau^2)}{2}} \ e^{\frac{2\xi^2}{1-\tau^2}} \, \frac{1}{2\vert \xi \vert} \ \textup{erfc}\left( \sqrt{\frac{2}{1-\tau^2}} \, \vert \xi \vert \right) \Bigg] \left( 1- \frac{\delta^2}{(1+\tau)^2} \right) \ \Theta \left[ 1- \frac{\delta^2}{(1+\tau)^2} \right] \ .
    \end{align}
\end{cor}

\noindent
In Figure \ref{fig:depletion_SNH}, we plot the mean conditional self-overlap of eigenvectors associated with complex eigenvalues within the depletion regime of the eGinOE for a range of $\tau$. Here it can be seen that as $\tau$ increases (for a fixed $\xi$) the conditional self-overlap decreases. This reflects the fact that the ellipse is thinner for a larger $\tau$ and so the fixed value of $\xi$ will lie closer to the spectral bulk.

\begin{figure}[h!]
    \centering
    \includegraphics[scale = 0.3]{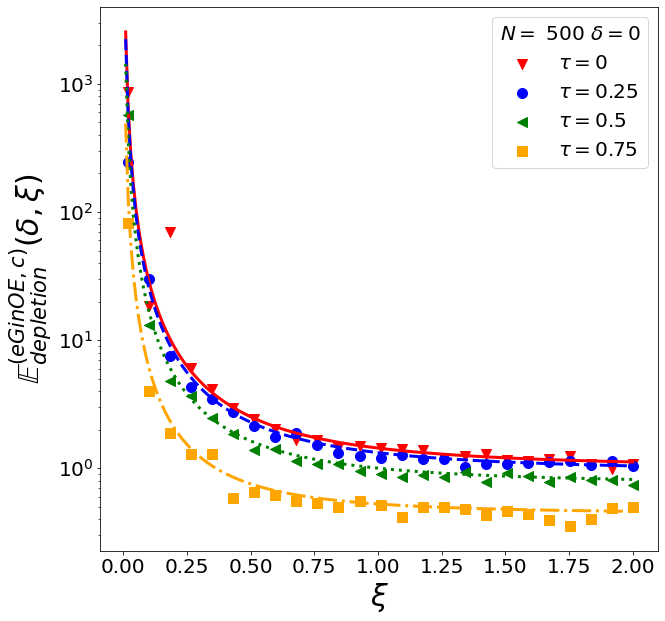}
    \hspace{2cm}
    \includegraphics[scale = 0.3]{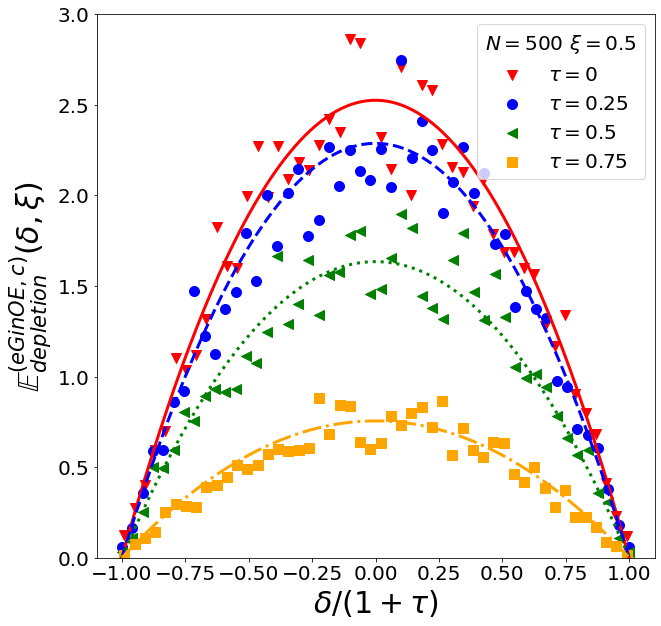}
    
    \caption{\small Mean conditional self-overlap, $\mathbb{E}^{\text{(eGinOE,c)}}_{\text{depletion}}(\delta,\xi)$, within the depletion regime of the eGinOE at large $N$ for a range of $\tau$. Left: $\mathbb{E}^{\text{(eGinOE,c)}}_{\text{depletion}}(\delta,\xi)$ measured as a function of $\xi$ at fixed $\delta=0$. Right: $\mathbb{E}^{\text{(eGinOE,c)}}_{\text{depletion}}(\delta,\xi)$ measured as a function of $\delta$ at fixed $\xi=0.5$. In both plots, theoretical predictions (lines) are compared to numerical simulations of $N=500$ eGinOE matrices (coloured markers). In simulations, for a required set of $\delta$ and $\xi$, mean self-overlap values are obtained by averaging self-overlaps associated with complex eigenvalues within a tolerance of $\pm 1/\sqrt{N}$ from a set containing $O(10^8)$ samples.}
    \label{fig:depletion_SNH}
\end{figure}

\begin{rem}
    As expected, setting $\tau=0$ in any of the above formulas recovers the GinOE or GinUE case considered in \cite{WCF}.
\end{rem}


\subsubsection{Weak non-Hermiticity Regime}
\label{subsubsec:WeakRegime}

\noindent
The weak non-Hermiticity regime is reached when as $N \to \infty$ the parameter $\tau$ approaches $1$, such that $\tau = 1-(\pi\alpha)^2/(2N)$. This describes the transition from non-Hermitian to Hermitian matrices, the GUE or GOE. Starting from our finite $N$ results for the mean self-overlap, one may derive the following expressions in the \emph{weak bulk} limit: 
\begin{cor}\label{cor:eGinUEbulkWeak}
    For a complex eigenvalue $z= \sqrt{N} X + (i \pi y)/\sqrt{N}$, the limiting scaled mean self-overlap of left and right eigenvectors in the weak bulk of the eGinUE, where $\tau = 1- (\pi \alpha)^2/(2N)$, reads
    \begin{equation} \label{Eq:eGinUEbulkWeakRes}
        \begin{split}
            \mathcal{O}_{\textup{bulk,WNH}}^{\textup{(eGinUE,c)}}(X,y) \equiv& \lim_{N\rightarrow \infty} \frac{\pi^2}{N} \ \mathcal{O}_N^{\textup{(eGinUE,c)}}\left( z =\sqrt{N} X + \frac{i \pi y}{\sqrt{N}} \right) \\
            =& \frac{\sqrt{2}}{\pi^{3/2}} \, \frac{1}{|\alpha|} \, e^{-\frac{2y^2}{\alpha^2}}  \, \int_{0}^{\pi\sqrt{1 - \frac{X^2}{4}}} du \ e^{-\frac{\alpha^2u^2}{2}} \cosh(2yu) \left[ 1 + \alpha^2 \pi^2 \left(1 - \frac{X^2}{4} - \frac{u^2}{\pi^2} \right) \right] \ .
        \end{split}
    \end{equation}
\end{cor}

\begin{cor}\label{cor:eGinOEbulkWeak}
    For a complex eigenvalue $z= \sqrt{N} X + (i \pi y)/\sqrt{N}$, with $y \neq 0$, the limiting scaled mean self-overlap of left and right eigenvectors in the weak bulk of the eGinOE, where $\tau = 1- (\pi \alpha)^2/(2N)$, reads 
    \begin{align}
        \mathcal{O}_{\textup{bulk, WNH}}^{\textup{(eGinOE,c)}}&(X,y) \equiv \lim_{N\rightarrow \infty} \frac{\pi^2}{N} \ \mathcal{O}_N^{\textup{(eGinOE,c)}}\left( z =\sqrt{N} X + \frac{i \pi y}{\sqrt{N}} \right) 
        = \frac{1}{\sqrt{2} \, \pi^{3/2}} \, \frac{|\alpha|}{y} \, e^{-\frac{2 y^2}{\alpha^2}} \label{Eq:eGinOEbulkWeakRes}   \\
        &\left[ 1 + \sqrt{\frac{\pi}{2}} \, \frac{|\alpha|}{2\vert y \vert} \, e^{\frac{2 y^2}{\alpha^2}} \, \textup{erfc}\left( \frac{\sqrt{2}\vert y \vert}{|\alpha|}  \right) \, \right] \int_{0}^{\pi \sqrt{1 - \frac{X^2}{4}} } du \, e^{-\frac{ \alpha^2 u^2}{2}} \, u \, \sinh{\left( 2 yu \right)} \left[ 1 + \alpha^2 \pi^2 \left(1 - \frac{X^2}{4} - \frac{u^2}{\pi^2} \right) \right] \ .  \nonumber
    \end{align}
\end{cor}

\noindent
We plot the mean conditional self-overlap for both the eGinOE and eGinUE at WNH in Figure \ref{fig:WNH}. This is done in two different ways: as a function of $y$ at fixed $X=0$ and as a function of $X$ at fixed $y=1$. In both cases, it can be seen that the theory is accurate for all $X$ and $y$ and that the eGinOE generally has a higher conditional self-overlap. This implies that our asymptotic analysis accurately describes the mean self-overlap in the weak bulk of the eGinOE and eGinUE, as well as in the depletion regime of the eGinOE. Thus, the weak bulk and depletion regime of the eGinOE at WNH are described by the same scaling of $y$. 

\begin{figure}[h!]
    \centering
    \includegraphics[scale = 0.3]{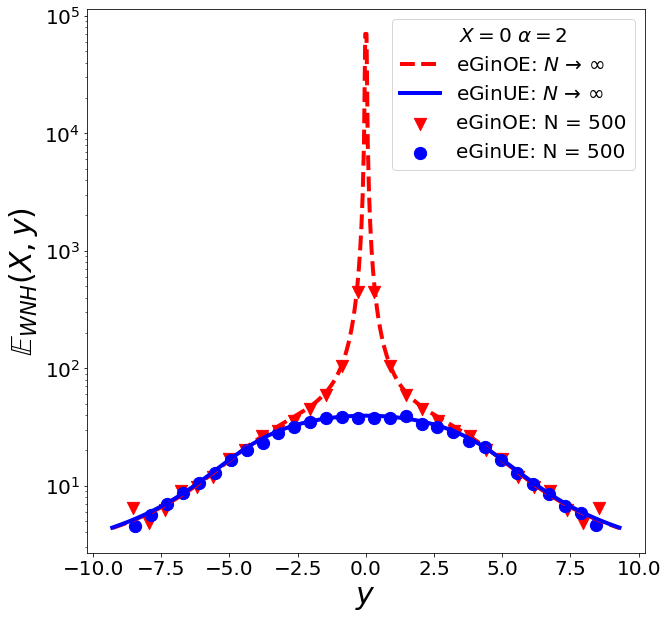}
    \hspace{2cm}
    \includegraphics[scale = 0.3]{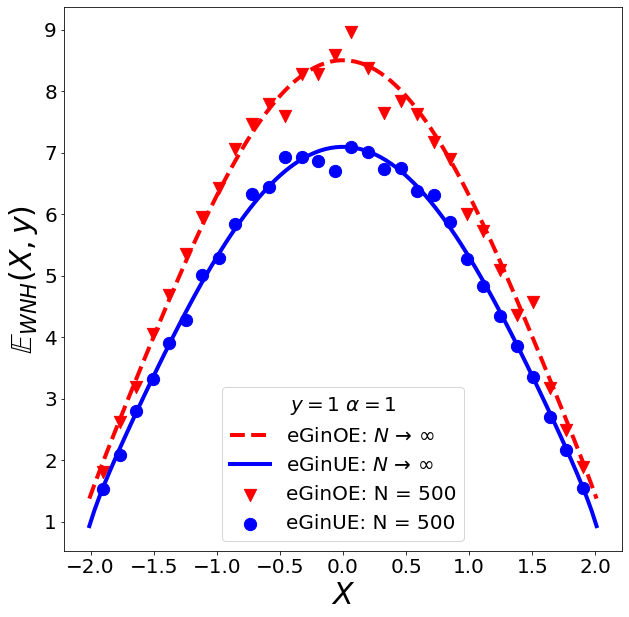}
    
    \caption{\small Mean conditional self-overlap of eigenvectors at WNH, $\mathbb{E}_{\text{WNH}}(X,y)$, in the eGinOE (red) and eGinUE (blue). In both plots our theory (lines) is compared to numerical simulations of matrices of size $N=500$ (markers). Left: $\mathbb{E}_{\text{WNH}}(X,y)$ measured as a function of $y$ for fixed $X=0$ with $\alpha =2$. Right: $\mathbb{E}_{\text{WNH}}(X,y)$ measured as a function of $X$ for fixed $y=1$ with $\alpha =1$. In simulations, for a given set of $X$ and $y$, mean values are calculated by averaging all self-overlaps associated with complex eigenvalues within $\pm 1/\sqrt{N}$ of the desired $X$ and $y$ from a sample containing $O(10^8)$ self-overlaps. }
    \label{fig:WNH}
\end{figure}

\subsection{Discussion of open problems} \label{subsec:OpenProbs}

In a subsequent paper, we will consider the mean self-overlap at the droplet edge for both strong and weak non-Hermiticity. We will also study the density of edge eigenvalues in the eGinOE, as (to the best of our knowledge) the leading order asymptotic behaviour has not yet been established in this ensemble. It is natural to expect that the result should match the leading order eGinUE result derived by Lee and Riser \cite{LR}, who also give fine asymptotic behaviour of the density in terms of the curvature of the ellipse. The leading order density term, as $N \rightarrow \infty$, is universal as described by Tao and Vu \cite{TaoVu}, and should also be valid in the eGinOE at the edge. This could in principle be derived using Forrester and Nagao's finite $N$ result. In the present paper we will not discuss the edge statistics of eigenvalue density or self-overlaps further.

Alternatively, one could discuss similar objects in the third standard Ginibre ensemble, the \emph{quaternionic Ginibre ensemble}, to connect with the existing results in \cite{AFK} and then extend to its elliptic version. For that purpose it might also be interesting to employ techniques similar to the works by Noda \cite{Noda23a,Noda23b} and Akemann et al \cite{ATTZ}.

Finally,  the ultimate goal should be to describe the entire distribution of diagonal overlaps at finite $N$ and in various asymptotic regimes. The distribution of self-overlaps associated with complex eigenvalues has been considered in the GinUE \cite{BD,FyodorovCMP} and has been studied for real eigenvalues in the eGinOE \cite{FT}, but is not yet available for complex eigenvalues in both eGinOE and eGinUE. Considering the statistics of off-diagonal overlaps, i.e. Eq. \eqref{eq:def_overlap} with $n \neq m$, beyond results of \cite{BD} is another important avenue.

\noindent


\section{Derivation of Main Results}\label{sec:MainResDerivation}

\noindent
To prove our main results, we review the Schur decomposition of complex and real random matrices as well as their effects on $\mathcal{O}(z)$ in Sections \ref{subsec:inSchurCOMPLEX} and \ref{subsec:inSchurREAL} respectively. Then, through the use of techniques from \cite{FyodorovCMP,FT,WCF}, we prove Theorems \ref{thm:MainResCOMPLEX} and \ref{thm:MainRes} in Sections \ref{subsec:ProoffiniteNCOMPLEX} and \ref{subsec:ProoffiniteN} respectively.


\subsection{Incomplete complex Schur decomposition}
\label{subsec:inSchurCOMPLEX}

\noindent
In the eGinUE, the eigenvalues are always complex and are not necessarily in conjugate pairs. Therefore, we make use of the \emph{incomplete complex Schur decomposition} (ICSD), see \cite{EKS,Edelman,FK}. Let $z=x+iy$ be a complex eigenvalue of the $N\times N$ complex, non-Hermitian matrix $X$, with associated right and left eigenvectors denoted by $\mathbf{x}_R$ and $\mathbf{x}_L^\dagger$ respectively. Note that to have a nontrivial self-overlap we must assume $N\geq 2$. Utilising the results derived in \cite{FK} we find the ICSD of $X$ with respect to $z$ as
\be\label{incomplSchurCOMPLEX}
    X = R \widetilde{X} R^\dagger \quad \text{with} \quad \widetilde{X} = \left(
    \begin{matrix}
    z & \mathbf{w}^\dagger \\
    0 & X_1 \\
    \end{matrix}
    \right)\text{,}
\ee
where $R$ is a Hermitian matrix with $R^\dagger = R$ and $R^2 = \eins_N$. The vector $\mathbf{w}$ is complex with $N-1$ entries with independent real and imaginary parts. The matrix $X_1$ is an $(N-1) \times (N-1)$ dimensional complex matrix with the same statistical properties as $X$, which in principle could be further decomposed into blocks. In  \cite[Section 4]{FyodorovCMP}, this decomposition is applied to the GinUE and it is found that a complex eigenvalue $z$, has right eigenvector $\bm v_R \rightarrow \left( 1, \bm 0_{N-1} \right)^T$ and left eigenvector $\bm v_L \rightarrow \left( 1, \mathbf{b}_{N-1}\right)^T$, with $N-1$ dimensional complex vector $\mathbf{b}_{N-1}$. We thus deduce that in the eGinUE the self-overlap, $\mathcal{O}_z$, is readily expressed using the resolvent of $X_1$, since $\mathbf{b}_{N-1}^\dagger = \mathbf{w}^\dagger \left(z\eins_{N-1} - X_1 \right)^{-1}$, implying
\be\label{Eq:bbdaggerCOMPLEX}
    \mathcal{O}_z = 1 + \mathbf{b}_{N-1}^\dagger \mathbf{b}_{N-1} = 1 + \mathbf{w}^\dagger \bigg[ \left(z\eins_{N-1} - X_1 \right)^\dagger \left(z\eins_{N-1} - X_1 \right) \bigg]^{-1} \mathbf{w} \equiv 1 + \mathbf{w}^\dagger \ B^{-1} \ \mathbf{w}\ ,
\ee
where we have defined the matrix 
\be \label{Bmat}
B \equiv  \left(z\eins_{N-1} - X_1 \right)^\dagger \left(z\eins_{N-1} - X_1 \right).
\ee
The Jacobian needed for the change to the ICSD of $X$ can be found in Appendix B of \cite{FK},
\be\label{JacSchur2}
\begin{split}
    dX &= \frac{1}{2} \ \vert \det \left[ z\eins_{N-1} - X_1 \right] \vert^2 \  d^2 z \ d^2\mathbf{w} \ dX_1 \ dR \ ,
\end{split}
\ee
where $dR$ denotes the contribution to the measure originating from the Householder reflection matrix $R$. After simple manipulations we find the probability measure, defined in terms of the new variables:
\be\label{JPDFzellipticCOMPLEX}
\begin{split}
    P_{\text{eGinUE}}(X)dX &= D_{N,\tau}^\prime \vert \det \left[ z\eins_{N-1} - X_1 \right] \vert^2 e^{-\frac{1}{1-\tau^2} \big[\Tr \left( X_1 X_1^\dagger - \tau \ \RE X_1^2 \right) + \left( \vert z \vert^2 - \tau \ \RE z^2 + \mathbf{w}^\dagger \mathbf{w} \right) \big]} d^2 z \ d^2\mathbf{w} \ dX_1 \ .
\end{split}
\ee
Note that the integral over $dR$ yields a factor $\frac{2\pi^{N-1}}{\Gamma(N)}$ \cite[Section 6]{FK}, combining this with $D_{N,\tau}$ from Definition \ref{def:eGinUE}, gives the new constant 
\be\label{Eq:NEWconstsCOMPLEX}
\begin{split}
    D_{N,\tau}^\prime 
    &= \frac{1}{2} \, D_{N,\tau}^{-1} \, \frac{2\pi^{N-1}}{\Gamma(N)} = \frac{\pi^{-N^2+N-1}}{\Gamma(N)} \, \left(1-\tau^2 \right)^{-N^2/2} \ .
\end{split}
\ee



\subsection{Proof of Theorem \ref{thm:MainResCOMPLEX} - eGinUE}\label{subsec:ProoffiniteNCOMPLEX}

\noindent
The mean self-overlap for the eGinUE can be written as
\be\label{Eq:OverlapsCOMPLEX}
\begin{split}
     \mathcal{O}^{(\eGinUE,c)}_{N}(z) &= \bigg\langle \frac{1}{N}\sum_{n=1}^N \mathcal{O}_{nn} \ \delta(z-z_n) \bigg\rangle_{\text{eGinUE},N} = \bigg\langle \mathcal{O}_{\widetilde{z}} \ \delta(z-\widetilde{z}) \bigg\rangle_{X} \ .
\end{split}
\ee
The ICSD of an $N \times N$ eGinUE matrix $X$ allows one to express $\mathcal{O}_z$ in terms of $z$ and $X_1$, thus the average is obtained by integrating over the measure in Eq. \eqref{JPDFzellipticCOMPLEX}. The $\delta$-function of a complex variable makes the integration over $d^2 z$ trivial. It is convenient to define the normalised average of any function $\mathcal{A}(\mathbf{w})$ with respect to the vector $\mathbf{w}$ as
\be\label{Eq:WaverageDefCOMPLEX}
    \bigg\langle \mathcal{A}(\mathbf{w}) \bigg\rangle_{\mathbf{w}} \equiv  \frac{(1-\tau^2)^{-N+1}}{\pi^{N-1}} \int d^2 \mathbf{w} \exp \left[ -\frac{1}{1-\tau^2} \mathbf{w}^\dagger \mathbf{w} \right] \mathcal{A}(\mathbf{w}) \ ,
\ee 
so that we can write
\begin{align}
    \bigg\langle \mathcal{O}_{\widetilde{z}} \ \delta(z-\widetilde{z}) \bigg\rangle_{X} 
    &= D_{N,\tau}^\prime \ e^{-\frac{1}{1-\tau^2}\left( \vert z \vert^2 - \tau \ \RE z^2 \right)} \int dX_1 \  \vert \det \left[ z\eins_{N-1} - X_1 \right] \vert^2 \ e^{-\frac{1}{1-\tau^2}\Tr \left( X_1 X_1^\dagger - \tau \ \RE X_1^2 \right) } \nonumber\\
    &\times \ (1-\tau^2)^{N-1} \ \pi^{N-1} \  \bigg\langle 1 +   \left( \mathbf{b}_{N-1}^\dagger \mathbf{b}_{N-1} \right) \bigg\rangle_{\mathbf{w}} \ .
    \label{Eq:avOvleGinUEstep}
\end{align}
We can now adapt statements from \cite{WCF}, in particular Lemmas 3.1 and 3.2 which apply to the GinOE. These adaptations are straightforward exercises in Gaussian integrals and are left to the reader:

\begin{lem}\label{lem:Wav1COMPLEX}
Let $\mathbf{w}$ be a complex vector of length $N-1$, whose entries have independent real and imaginary parts. We have
\be\label{Eq:LemWaverage1COMPLEX}
\begin{split}
    &\bigg\langle \mathbf{w}^\dagger Y \mathbf{w} \bigg\rangle_{\mathbf{w}} = (1-\tau^2) \Tr Y \ ,
\end{split}
\ee
where the average is defined in Eq. \eqref{Eq:WaverageDefCOMPLEX} and $Y$ is an $N-1$ dimensional square matrix.
\end{lem}

\noindent
Straightforward application of this formula further implies

\begin{lem}\label{lem:Wav2COMPLEX}
    The average over $\mathbf{w}$ of $\mathbf{b}_{N-1}^\dagger \mathbf{b}_{N-1} $, defined in Eq. \eqref{Eq:bbdaggerCOMPLEX}, is given by
    \be\label{Eq:LemWaverage2COMPLEX}
    \begin{split}
        \bigg\langle 1 +   \left( \mathbf{b}_{N-1}^\dagger \mathbf{b}_{N-1} \right) \bigg\rangle_{\mathbf{w}} 
        &=  1 +(1-\tau^2) \Tr \left[ B^{-1} \right] \ .
    \end{split}
    \ee
\end{lem}

\noindent
Inserting now the result from Lemma \ref{lem:Wav2COMPLEX} into Eq. \eqref{Eq:avOvleGinUEstep}, the next task is to perform the integral over $X_1$ in
\be\label{Eq:avOvleGinUEstep2}
\begin{split}
    \bigg\langle \mathcal{O}_{\widetilde{z}} \ \delta(z-\widetilde{z}) \bigg\rangle_{X} 
    &= D_{N,\tau}^\prime \ e^{-\frac{1}{1-\tau^2}\left( \vert z \vert^2 - \tau \ \RE z^2 \right)} \ (1-\tau^2)^{N-1} \ \pi^{N-1} \\
    &\times   \int dX_1 \  \vert \det \left[ z\eins_{N-1} - X_1 \right] \vert^2 \ e^{-\frac{1}{1-\tau^2}\Tr \left( X_1 X_1^\dagger - \tau \ \RE X_1^2 \right) } \bigg[ 1 + (1-\tau^2) \Tr B^{-1} \bigg] \ .
\end{split}
\ee
We now utilise the definition of the matrix $B$ in Eq. (\ref{Bmat}) and notice that we can write
\be\label{detformulasCOMPLEX}
\begin{split}
    \det B 
    &=  \vert \det \left[ z\eins_{N-1} - X_1 \right] \vert^2  = \det \left[ \begin{matrix}
    0 & i(z\eins_{N-1} - X_1)  \\
    i(\bar{z} \eins_{N-1} - X_1^\dagger) & 0 \\
    \end{matrix}  \right] \ ,
\end{split}
\ee
which will be useful in calculating the first term of Eq. \eqref{Eq:avOvleGinUEstep2}, whereas for the second term we must utilise
\be\label{deriv}
    \frac{\partial}{\partial \mu} \det\left( \mu \eins_{N-1} + B \right) \bigg\vert_{\mu=0} = \det B \ \Tr B^{-1} \ .
\ee
This motivates the need to introduce the block-matrix $M$ via
\be\label{Eq:propM2mat}
    M \equiv  \left( \begin{matrix}
        \sqrt{\mu} \eins_{N} & i\left( z \eins_N - X \right) \\
        i\left( \bar{z} \eins_N - X^\dagger \right) & \sqrt{\mu} \eins_N \\
    \end{matrix} \right) \ ,
\ee
where each block is of size $N\times N$ and $\mu$ is a real parameter, such that for $N\rightarrow N-1$ and $\mu=0$ we have
\be\label{detMpartialdetM}
    \det M \ \bigg\vert_{\mu=0} = \det B  \quad \text{and} \quad \frac{\partial}{\partial \mu} \det M \ \bigg\vert_{\mu=0} = \det B \ \Tr B^{-1} \ .
\ee
Now we define the normalised average of any function $\mathcal{A}(X_1)$ over $X_1$ via
\be\label{Eq:X2avCOMPLEX}
    \bigg\langle \mathcal{A}(X_1) \bigg\rangle_{X_1} \equiv  \frac{1}{D_{N-1,\tau}} \int dX_1 \ \exp \left[-\frac{1}{1-\tau^2}\Tr \left( X_1 X_1^\dagger - \tau X_1^2 \right) \right] \ \mathcal{A}(X_1) \ ,
\ee
where the constant $D_{N-1,\tau}$ is given in Definition \ref{def:eGinUE}. Introducing another constant, $D_{N,\tau}^{\prime \prime} \equiv D_{N,\tau}^\prime D_{N-1,\tau} (1-\tau^2)^{N-1} \ \pi^{N-1}$, the expression in Eq. \eqref{Eq:avOvleGinUEstep2} takes the form
\be\label{Eq:avOvleGinUEstep3}
\begin{split}
     \bigg\langle \mathcal{O}_{\widetilde{z}} \ \delta(z-\widetilde{z}) \bigg\rangle_{X}  
     &= D_{N,\tau}^{\prime \prime} \ e^{-\frac{1}{1-\tau^2}\left( \vert z \vert^2 - \tau \ \RE z^2 \right)}  \bigg[ \  \bigg\langle \det M \bigg\rangle_{X_1} \, \bigg\vert_{\mu = 0} + (1-\tau^2) \frac{\partial}{\partial \mu }\bigg\langle \det M \bigg\rangle_{X_1} \, \bigg\vert_{\mu = 0} \ \bigg] \ .
\end{split}
\ee
The averaging of $\det M$ with respect to the eGinUE (and later the eGinOE) can be most conveniently performed using the representation of the determinant via Berezin integration over anti-commuting Grassmann variables:
\be\label{detIdentGrass}
    \det M  = \int D(\bm \Phi,\bm \chi) \exp \left[ - \bm \Phi^T M \bm \chi \right]  \ \text{,}
\ee
where $D(\bm \Phi,\bm \chi) = d\bm \phi_1 d \bm \chi_1 d \bm \phi_2 d\bm \chi_2$ and $\bm \Phi^T = \left( \bm \phi_1, \bm \phi_2 \right)^T$, $\bm \chi^T = \left( \bm \chi_1, \bm \chi_2 \right)^T$. Here, $\bm \phi_1$, $\bm \phi_2$, $\bm \chi_1$ and $\bm \chi_2$ are Grassmann vectors of length $N$. For further use it is also convenient to introduce a longer Grassmann vector $\phi^T \equiv \left( \bm \phi_1,\bm \phi_2, \bm \chi_1, \bm \chi_2 \right)^T$ with the measure $D\phi =  d \bm \phi_1 d \bm \phi_2  d\bm \chi_1  d \bm \chi_2$, as this allows us to employ the following identity, see e.g.  \cite{CSS}
\be\label{PfaffIdentGrass}
    \int D\phi \ \exp \left[ - \frac{1}{2} \phi^T H \phi \right] = 
    \begin{cases}
        \Pf H & \text{for } H \ \text{even dimensional,} \\
        0 & \text{for } H \ \text{odd dimensional.} \\
    \end{cases}
\ee
The outcome of calculating the average determinant of $M$ over $X$ is provided in the following Proposition.
\begin{prop}\label{prop:Xav}
    Let $M$ be the matrix defined in Eq. \eqref{Eq:propM2mat}, where $z$ is a complex number and $X$ is an $N\times N$ eGinUE matrix. The average of $\det M$, with respect to $X$, can be expressed in the following form:
    \begin{align}
        \bigg\langle \det M \bigg\rangle_{X} 
        &= \frac{1}{2\pi} \int_{\mathbb{R}^2} dv_1 \, dv_2 \, e^{ -\frac{v_1^2   + v_2^2}{2} }  \int_0^{\infty} dR  \ e^{-R}  \bigg[  \Big(R + f(z,v_1,v_2) \Big)^2 + \mu^2 +2\mu \Big(f(z,v_1,v_2) -R\Big)  \bigg]^{\frac{N}{2}} \nonumber\\
        &   \times P_N\left(\frac{\mu+R+f(z,v_1,v_2) }{\sqrt{\big(R + f(z,v_1,v_2) \big)^2 + \mu^2 +2\mu \big(f(z,v_1,v_2) -R\big)}}\right) \ , \label{Eq:detMreseGinUEb}
    \end{align}
    where $f(z,v_1,v_2) \equiv \left( \bar{z} -i v_2 \sqrt{\tau} \right) \left( z +i v_1 \sqrt{\tau} \right)$ and $P_N(t) $ is a Legendre polynomial defined in \textup{\cite[8.913.3]{Grad}}
    \be \label{LegPolyint}
        P_N(t)=\frac{1}{\pi}\int_0^{\pi} d\theta \ \left(t+\sqrt{t^2-1}\cos{\theta}\right)^N \ .
    \ee
\end{prop}

\begin{proof}
     Using Eq. \eqref{detIdentGrass} to express the determinant of $M$ in terms of anticommuting variables leads to
    \be
        \bigg\langle \det M \bigg\rangle_{X} = (-1)^N\expval{ \int D \bm \Phi D \bm \chi \exp{ -  \left( \begin{matrix} \bm \phi_1^T & \bm \phi_2^T \end{matrix} \right) \left( \begin{matrix} \sqrt{\mu} \eins_N & i \left( z \eins_N - X \right) \\ i \left( \bar{z} \eins_N - X^\dagger \right) & \sqrt{\mu} \eins_N 
        \end{matrix} \right) \left( \begin{matrix} \bm \chi_1 \\ \bm \chi_2 \end{matrix} \right)
        }  }_{X} \ ,
    \ee
    which upon using $\bm \phi_i^T X \bm \chi_j = - \Tr \left[ X \bm \chi_j \bm \phi_i^T \right]$ can be written as
    \begin{align}
        \bigg\langle \det M \bigg\rangle_{X} &= (-1)^N \int D \bm \Phi D \bm \chi \ e^{ - \sqrt{\mu} \left( \bm \phi_1^T \bm \chi_1 + \bm \phi_2^T \bm \chi_2 \right) - i \left(z  \bm \phi_1^T \bm \chi_2  + \bar{z} \bm \phi_2^T \bm \chi_1 \right) }  \expval{  e^ {- i \Tr[ X \bm \chi_2 \bm \phi_1^T] - i \Tr[ X^\dagger \bm \chi_1 \bm \phi_2^T]}
         }_{X}  \ . \label{eq:detM_eGinUE_ave}
    \end{align}
    We now employ the following identity for averages with respect to the eGinUE:
    \begin{equation}
        \expval{e^{-\Tr[X A + X^\dagger B] }}_{X} = \exp[  \Tr \left( A B \right) + \frac{\tau}{2} \Tr \left( A^2 + B^2 \right)   ] \ ,
        \label{eq:ave_exp_Tr_GNell2}
    \end{equation}
    where if we use $A = i \bm \chi_2 \bm \phi_1^T$ and $B = i \bm \chi_1 \bm \phi_2^T$, the average in Eq. \eqref{eq:detM_eGinUE_ave} becomes
    \be\label{EqXell2}
    \begin{split}
        \expval{e^{-\Tr[X A + X^\dagger  B] }}_{X} &= \exp \bigg[ \left( \bm \phi_1^T \bm \chi_1 \right) \left( \bm \phi_2^T \bm \chi_2 \right) + \frac{\tau}{2} \bigg( \left( \bm \phi_1^T \bm \chi_2 \right)^2  +  \left( \bm \chi_1^T \bm \phi_2 \right)^2  \bigg)  \bigg] \ . 
    \end{split}
    \ee
    These terms, quartic in Grassmann variables, can be further bilinearized by employing auxilliary Gaussian integrals, i.e. Hubbard-Stratonovich transformations:
    \be\label{ellHubbStratauxCOMPLEX}
    \begin{split}
        \exp \bigg[ \left(\bm \phi_1^T \bm \chi_1 \right) \left( \bm \phi_2^T \bm \chi_2 \right) \bigg] &= \frac{1}{2\pi} \int_{\mathbb{C}} dq d\bar{q} \ \exp \bigg[ -\vert q \vert^2 - q \left( \bm \phi_2^T \bm \chi_2 \right) - \bar{q}\left(\bm \phi_1^T \bm \chi_1 \right) \bigg] \quad \text{,}\\
        \exp \bigg[ \ \frac{\tau}{2}\left( \bm \phi_1^T \bm \chi_2 \right)^2 \ \bigg] &= \frac{1}{\sqrt{2\pi}} \int_{\mathbb{R}} dv_1 \ \exp \bigg[-\frac{1}{2}\ v_1^2 + v_1 \ \sqrt{\tau} \ \left( \bm \phi_1^T \bm \chi_2 \right) \bigg] \quad \text{,}\\
        \exp \bigg[ \ \frac{\tau}{2}\left( \bm \chi_1^T \bm \phi_2 \right)^2 \ \bigg] &= \frac{1}{\sqrt{2\pi}} \int_{\mathbb{R}} dv_2 \ \exp \bigg[-\frac{1}{2}\ v_2^2 + v_2 \ \sqrt{\tau}\ \left( \bm \chi_1^T \bm \phi_2 \right) \bigg] \quad \text{.}
    \end{split}
    \ee
    Rewriting the exponents using that $\phi \chi = (\phi \chi - \chi \phi)/2$, for anti-commuting variables $\chi$ and $\phi$, allows the usage of the Pfaffian as defined in Eq. \eqref{PfaffIdentGrass} with an anti-symmetric matrix $H$ of the form
    \be\label{MPfaffMatCOMPLEX}
        H = \left(\begin{matrix}
        0 \ \eins_N &  0 \ \eins_N & \left( \sqrt{\mu} + \bar{q} \right) \eins_{N} & \left( i z - v_1 \sqrt{\tau} \right) \eins_{N} \\
        0 \ \eins_N  & 0 \ \eins_N & \left(  i \bar{z} + v_2 \sqrt{\tau} \right) \eins_{N} & \left( \sqrt{\mu} + q \right) \eins_{N} \\
        -\left( \sqrt{\mu} + \bar{q} \right) \eins_{N} & -\left(  i \bar{z} + v_2 \sqrt{\tau} \right) \eins_{N} & 0 \ \eins_N & 0 \ \eins_N \\
        -\left( i z - v_1 \sqrt{\tau} \right) \eins_{N} & -\left( \sqrt{\mu} + q \right) \eins_{N} &  0 \ \eins_N & 0 \ \eins_N \\
        \end{matrix} \right) = \left( \begin{matrix}
        0 & \widetilde{H} \\
        -\widetilde{H}^T & 0 \\
    \end{matrix} \right) \ .
    \ee
    Due to the structure of $H$ we can use an identity in \cite{DVarjas} to write  $\Pf H = (-1)^N \det \widetilde{H}$, with
    \be\label{MPfaffMat2COMPLEX}
        \widetilde{H} = F \otimes \eins_N =  \left(\begin{matrix}
        \sqrt{\mu} + \bar{q}  &  i z - v_1 \sqrt{\tau}  \\
        i \bar{z} + v_2 \sqrt{\tau} & \sqrt{\mu} + q  \\
        \end{matrix} \right) \otimes \eins_N \ .
    \ee
    The determinant obeys the rule $\det \left( A \otimes B \right) = \left( \det A\right)^{m} \left( \det B \right)^{n}$ for matrices $A$ and $B$ of dimension $n$ and $m$ respectively. In our case $A=F$ and $B=\eins_N$ thus $n=2$ and $m=N$, therefore $\det \widetilde{H} = \left( \det F\right)^{N}$. The determinant of $F$ can be immediately calculated yielding
    \be\label{MPfaffMat4COMPLEX}
        \Pf H = (-1)^N \left( \det F \right)^N = (-1)^N \bigg[ \left( \sqrt{\mu} + \bar{q} \right) \left( \sqrt{\mu} + q \right)  + \left( \bar{z} -i v_2 \sqrt{\tau} \right) \left( z +i v_1 \sqrt{\tau} \right)  \bigg]^N
    \ee
    and after straightforward manipulations our equation for the average determinant of $M$ becomes
    \begin{align}
        \bigg\langle \det M \bigg\rangle_{X}  
        &=(-1)^N \left( \frac{1}{2\pi} \right)^2 \int_{\mathbb{C}} dq d\bar{q}  \int_{\mathbb{R}^2} dv_1 dv_2 \exp \bigg[  -\vert q \vert^2 -\frac{1}{2} \bigg( v_1^2   + v_2^2 \bigg) \bigg]  \ \Pf H \nonumber \\
        &= \frac{1}{4\pi^2} \int_{\mathbb{C}} dq d\bar{q}  \ \int_{\mathbb{R}^2} dv_1 \,  dv_2 \,  e^{-\vert q \vert^2  -\frac{v_1^2 + v_2^2}{2}}  \bigg[ \left( \sqrt{\mu} + \bar{q} \right) \left( \sqrt{\mu} + q \right)   + f(z,v_1,v_2)  \bigg]^N \ . \label{MeEllres1COMPLEX}
    \end{align}
    Employing the change of variables, $q = \sqrt{R} e^{i\theta}$ and $\bar{q} = \sqrt{R} e^{-i\theta}$ along with the definition of a Legendre polynomial in Eq. \eqref{LegPolyint} we finally arrive at Eq. \eqref{Eq:detMreseGinUEb}.
\end{proof}

\noindent
The result derived in Proposition \ref{prop:Xav} is crucial to our analysis of $\langle \det M\rangle_{X_1}$ and its derivative w.r.t. $\mu$ at $\mu =0$. The results of this are shown in the following Corollary.

\begin{cor}\label{cor:AvdetX1}
    With the average taken as in Eq. \eqref{Eq:X2avCOMPLEX} using the $N-1$ dimensional eGinUE matrix $X_1$ and matrix $M$ given Eq. \eqref{Eq:propM2mat}, we have
    \be\label{Eq:AvdetMemu0COMPLEX}
        \bigg\langle \det M \bigg\rangle_{X_1} \ \bigg\vert_{\mu=0} = \Gamma(N) \, \pi \sqrt{1-\tau^2}\exp \left[ \frac{1}{1-\tau^2}\bigg( \vert z \vert^2 - \tau \ \textup{Re}(z^2) \bigg) \right] \rho_{N}^{\textup{(eGinUE,c)}}(z) \ ,
    \ee
    \be\label{Eq:AvdetMepartialmu0COMPLEX}
    \begin{split}
         \frac{\partial}{\partial \mu} \bigg\langle \det M \bigg\rangle_{X_1} \bigg\vert_{\mu=0} &= \Gamma(N) \pi \sqrt{1-\tau^2}e^{ \frac{\vert z \vert^2 - \tau \textup{Re}(z^2)}{1-\tau^2}} \bigg[ \rho_{N-1}^{\textup{(eGinUE,c)}}(z) + (N-2)  \rho_{N-2}^{\textup{(eGinUE,c)}}(z)   -  R_{N-3} \bigg] \ ,
    \end{split}
    \ee
    where $\rho_{N}^{\textup{(eGinUE,c)}}(z)$ and $R_N$ are defined in Eq. \eqref{eq:eGinUE_density}  and Eq. \eqref{Eq:RNRes} respectively and $\Gamma(N) = (N-1)!$. 
\end{cor}

\begin{proof}
    We follow the steps performed in \cite{WCF} to prove these statements. To derive Eq. \eqref{Eq:AvdetMemu0COMPLEX} we directly use Proposition \ref{prop:Xav}, for $\mu =0$ with $N\rightarrow N-1$ and the Legendre polynomial property, $P_{N}(1)=1$, producing
    \begin{align}
        \bigg\langle \det M \bigg\rangle_{X_1} \ \bigg\vert_{\mu=0} 
        &= \frac{1}{2\pi} \int_{\mathbb{R}^2} dv_1 dv_2 \ e^{-\frac{1}{2} \left( v_1^2   + v_2^2 \right)}  \int_0^{\infty} dR  \ e^{-R}   \big[R + f(z, v_1, v_2) \big]^{N-1} \ .
        \label{Eq:AvdetmuCOMPLEX}
    \end{align}
    Now employing the binomial theorem  one notices the connection to the Hermite polynomials utilizing the first equality in Eq. \eqref{Eq:HermitePoly}, which upon performing the integral over $R$ gives
    \begin{align}
        \bigg\langle \det M \bigg\rangle_{X_1} \ \bigg\vert_{\mu=0}&= \frac{1}{2\pi} \sum_{k=0}^{N-1} \left( \begin{matrix}
        N - 1 \\  k\\
        \end{matrix} \right) \int_0^{\infty} dR  \ e^{-R} R^{N-1-k} \int_{\mathbb{R}^2} dv_1 dv_2 \ e^{-\frac{1}{2} \left( v_1^2   + v_2^2 \right)} \Big( z + i v_1 \sqrt{\tau} \Big) \Big( \bar{z} - i v_2 \sqrt{\tau} \Big)^k \nonumber \\
        &= (N-1)! \sum_{k=0}^{N-1} \ \frac{\tau^k}{k!} \   \HE_k\left( \frac{z}{\sqrt{\tau}} \right)  \ \HE_k\left( \frac{\bar{z}}{\sqrt{\tau}} \right) \ .
        \label{Eq:AvdetMmu0StepCOMPLEX}
    \end{align}
    Recalling the definition of the density in Eq. \eqref{eq:eGinUE_density}, we obtain
    \be
        \bigg\langle \det M \bigg\rangle_{X_1} \ \bigg\vert_{\mu=0} = (N-1)! \, \pi \sqrt{1-\tau^2}\exp \left[ \frac{1}{1-\tau^2}\bigg( \vert z \vert^2 - \tau \ \text{Re}(z)^2 \bigg) \right] \rho_{N}^{\textup{(eGinUE,c)}}(z)
    \ee
    as desired for the first relation in Corollary \ref{cor:AvdetX1}. For the second relation, we start with Eq. \eqref{Eq:detMreseGinUEb}, replace $N\rightarrow N-1$, differentiate over $\mu$ using the formula \cite[Eq. (3.46)]{WCF}
    \be\label{diffLegend}
        \frac{\partial}{\partial \mu} \left( \frac{\mu+R+f(z,v_1,v_2) }{\sqrt{\big(R + f(z,v_1,v_2) \big)^2 + \mu^2 +2\mu \big(f(z,v_1,v_2) -R\big)}} \right) \Bigg \vert_{\mu = 0} = \frac{2R}{\Big(R + f(z, v_1,v_2) \Big)^2}
    \ee
    and employ $P_{N-1}^\prime (1)= N (N-1)/2$, which follows per induction from \cite[8.939.6]{Grad}. Collecting all terms we find
    \begin{align}
        &\frac{\partial}{\partial \mu} \bigg\langle \det M \bigg\rangle_{X_1} \, \bigg\vert_{\mu=0} = \frac{N-1}{2\pi}  \int_{\mathbb{R}^2} dv_1 \ dv_2 \ e^{-\frac{v_1^2   + v_2^2}{2}}  \int_0^\infty dR \  e^{-R} \  \Big( R + f(z,v_1,v_2) \Big)^{N-3} \big[ (N-1) R + f(z,v_1,v_2) \big] \nonumber \\
        &= \frac{N-1}{2\pi}  \int_{\mathbb{R}^2} dv_1  dv_2 \ e^{-\frac{v_1^2   + v_2^2}{2}} \int_0^\infty dR \  e^{-R} \   \bigg[ (R + f(z,v_1,v_2))^{N-2}  + (N-2) (R + f(z,v_1,v_2))^{N-3} \, R   \bigg] \ .
    \end{align}
    From Eqs. \eqref{Eq:AvdetMemu0COMPLEX} and \eqref{Eq:AvdetmuCOMPLEX}, we immediately identify that the first of the two terms above is proportional to $\rho_{N-1}^{\textup{(eGinUE,c)}}(z)$. The second term can be rewritten via the binomial theorem and the definition of the Hermite polynomials in Eq. \eqref{Eq:HermitePoly}, allowing us to obtain
\be
\begin{split}
	&\frac{(N-1)(N-2)}{2\pi}  \int_{\mathbb{R}} dv_1 \ \int_{\mathbb{R}} dv_2 \ e^{-\frac{1}{2} \left( v_1^2   + v_2^2 \right)} \int_0^\infty dR \ R \ e^{-R} \    \Big(R + f(z, v_1, v_2) \Big)^{N-3} \\
	&= \frac{(N-1)(N-2)}{2 \pi} \ \sum_{k=0}^{N-3} \left( \begin{matrix}
		N-3 \\		k\\
	\end{matrix} \right) \int_0^\infty dR e^{-R} R^{N-2-k} \int_{\mathbb{R}^2} dv_1 \, dv_2 e^{-\frac{v_1^2   + v_2^2}{2} } \Big( z + i \sqrt{\tau} v_1 \Big)^k \Big( \bar{z} - i \sqrt{\tau} v_2 \Big)^k \\
	&= (N-1)! \sum_{k=0}^{N-3} \frac{\tau^k}{k!} (N - 2 - k) \HE_k \left( \frac{z}{\sqrt{\tau}} \right) \HE_k \left( \frac{\bar{z}}{\sqrt{\tau}} \right) \ ,
\end{split}
\ee
which is proportional to $(N-2)\rho_{N-2}^{\textup{(eGinUE,c)}}(z)-R_{N-3}$. Combining the two terms we finally find
\be\label{res3ell2COMPLEX}
\begin{split}
    \frac{\partial}{\partial \mu} \bigg\langle \det M \bigg\rangle_{X_1} \ \bigg\vert_{\mu=0} &= \Gamma(N) \pi \sqrt{1-\tau^2} e^{ \frac{\vert z \vert^2 - \tau \ \textup{Re}(z)^2}{1-\tau^2}}  \bigg[ \rho_{N-1}^{\textup{(eGinUE,c)}}(z) + (N-2) \rho_{N-2}^{\textup{(eGinUE,c)}}(z) - R_{N-3}  \bigg] \ ,
\end{split}
\ee
as required.
\end{proof}

\noindent
To finish the proof of Theorem \ref{thm:MainResCOMPLEX} we insert our results for the averages over $X_1$ from Corollary \ref{cor:AvdetX1} into Eq. \eqref{Eq:avOvleGinUEstep3}. After some simple manipulation of the constants one arrives at the required result of
\begin{equation}
    \mathcal{O}^{(\eGinUE,c)}_{N}(z)  =  \rho_{N}^{\textup{(eGinUE,c)}}(z) + (1-\tau^2)\bigg[ \rho_{N-1}^{\textup{(eGinUE,c)}}(z) + (N-2) \rho_{N-2}^{\textup{(eGinUE,c)}}(z) - R_{N-3}  \bigg] \ .
\end{equation}

\subsection{Incomplete real Schur decomposition}
\label{subsec:inSchurREAL}

Now we move onto our analysis of the real-valued eGinOE, where the eigenvalues may come either in complex conjugate pairs or be real numbers. The latter case is not considered in this work as it has already been extensively studied in \cite{FT}. To treat complex conjugate pairs we employ \emph{incomplete real Schur decomposition} (IRSD) following the influential paper by Edelman \cite{Edelman} and restate the necessary information below.\\

\noindent
Let $z=x+iy$ ($y\neq 0$) be a complex eigenvalue of the $N\times N$ real, non-Hermitian matrix $X$, with $N \geq 2$. We find the IRSD of $X$ with respect to $z$ in \cite{Edelman}, i.e.
\be\label{incomplSchur}
X = Q \widetilde{X} Q^T \quad \text{with} \quad \widetilde{X} = \left(
\begin{matrix}
	\begin{matrix}
		x & b \\
		-c & x \\
	\end{matrix} & W \\
	0 & X_2 \\
\end{matrix}
\right)\text{,} \ bc >0 \text{,}  \ b \geq c \text{,} \ y = \sqrt{bc} \ ,
\ee
where $Q$ is a real symmetric matrix with $Q^T = Q$ and $Q^2 = \eins_N$. The block matrix $W$ is real and has size $2\times (N-2)$, i.e. $W= \left(\mathbf{w}_1 \ \text{,} \ \mathbf{w}_2 \right)^T$, where $\mathbf{w}_1$, $\mathbf{w}_2$ are vectors with $N-2$ real independent entries. The matrix $X_2$ is an $(N-2) \times (N-2)$ dimensional real matrix obeying the same statistical properties as $X$. This, in principle could be further decomposed into blocks \cite{Edelman}. As the eGinOE is an extension of the GinOE, the derivation of the representation of the left and right eigenvectors is exactly the same in both ensembles. Thus, we reuse results from \cite{WCF}, where the case of complex eigenvalues in the GinOE was considered. Specifically we utilise their Eqs. (3.5) to (3.9) to specify the self-overlap as
\be\label{overlapcomp}
\begin{split}
	\mathcal{O}_z &= \left( {\mathbf{x}}_L^\dagger {\mathbf{x}}_L \right) \left( {\mathbf{x}}_R^\dagger {\mathbf{x}}_R \right)
	=\frac{1}{4} \left( 2+ \frac{c^2+b^2}{bc} \right)+ \frac{1}{2} \left(1+\frac{c}{b} \right) \left( \mathbf{b}_{N-2}^\dagger \mathbf{b}_{N-2} \right)\ , 
\end{split}
\ee
where $\mathbf{b}_{N-2}$ is given by
\be\label{bcond}
    \mathbf{b}_{N-2}^\dagger = \frac{1}{\sqrt{2}}  \left(\mathbf{w}_1^T - i \sqrt{\frac{b}{c}} \mathbf{w}_2^T \right) \left(z\eins_{N-2} - X_2 \right)^{-1} \ .
\ee
Similarly to the previous section one can define the matrix 
\be\label{Bmat1}
    B \equiv  \left(z\eins_{N-2} - X_2 \right)^\dagger \left(z\eins_{N-2} - X_2 \right)
\ee 
so as to write
\be\label{Eq:bbdaggereGinOE}
\begin{split}
	\left( \mathbf{b}_{N-2}^\dagger \mathbf{b}_{N-2} \right) &= \frac{1}{2} \mathbf{w}_1^T \ B^{-1} \ \mathbf{w}_1 +\frac{1}{2} \exp \left[ 2 \ \text{arcsinh}\left( \frac{\delta}{2y} \right) \right] \mathbf{w}_2^T \  B^{-1} \ \mathbf{w}_2 \\
	&+\frac{1}{2} \ i \ \exp \left[ \text{arcsinh}\left( \frac{\delta}{2y} \right) \right]  \bigg[ \mathbf{w}_1^T \  B^{-1} \ \mathbf{w}_2 - \mathbf{w}_2^T \  B^{-1} \ \mathbf{w}_1 \bigg] \ .
\end{split}
\ee
Now  we change the variables from $b$ and $c$ to $y= \sqrt{bc}$ and $\delta = b-c>0$ which allows us to rewrite the mean self-overlap as
\be\label{overlapcomp2}
    \begin{split}
        \mathcal{O}_z 
        = \widetilde{c}_1 + \widetilde{c}_2 \   \left( \mathbf{b}_{N-2}^\dagger \mathbf{b}_{N-2} \right) \ ,
    \end{split}
\ee
where the constants are $y-$dependent and given by
\be\label{prefacs}
\begin{split}
	\widetilde{c}_1 &= \frac{1}{4} \left( 2+ \frac{\delta^2 + 2y^2}{y^2} \right)  \quad \text{and} \quad \widetilde{c}_2 =  \frac{1}{2} \left(1+ \exp \left[-2 \ \text{arcsinh}\left(\frac{\delta}{2y} \right) \right] \right) \ .
\end{split}
\ee
The equation for the Jacobian of the IRSD (before the last change of variables) can be found in \cite{Edelman}. For our purposes we integrate over the Stiefel manifold originating from the Householder reflection matrix $Q$. After straightforward manipulations, this yields the probability measure defined in terms of the new variables as
\begin{align}
	P_{\text{eGinOE}}(X)dX &= C_{N,\tau}^\prime \  \det \left[ (x\eins_{N-2} - X_2)^2 +y^2 \eins_{N-2} \right] \exp \left[ -\frac{1}{2(1-\tau^2)}\Tr \left( X_2 X_2^T - \tau X_2^2 + WW^T \right) \right]  \nonumber \\
	&\times \exp \left[ -\frac{x^2}{1+\tau} -\frac{y^2}{1-\tau} -\frac{\delta^2}{2(1-\tau^2)} \right] \ \frac{2y\delta}{\sqrt{\delta^2 + 4y^2}} \ dx \ dy \ d\delta \ dW \ dX_2 \ . \label{Eq:JPDF_eGinOE_IRSD}
\end{align}
The integral over the Stiefel manifold, taken from \cite[Eq. (20)]{Edelman}, combined with $C_{N,\tau}$ (see Definition \ref{def:eGinOE}) gives the new constant:
\be\label{Eq:NEWconsts}
\begin{split}
	C_{N,\tau}^\prime 
	&= 2\ \frac{(2\pi)^{-\frac{1}{2}(N-1)^2}}{\sqrt{2\pi} \, \Gamma(N-1)} \left( 1 + \tau \right)^{-\frac{N(N+1)}{4}} \left( 1- \tau \right)^{-\frac{N(N-1)}{4}} \ ,
\end{split}
\ee
where we have multiplied by the  factor of $2$ to count each complex eigenvalue $z=x+iy$ individually, for all $y\neq 0$. This is in contrast to some results in the literature, where only complex eigenvalue pairs $x\pm iy$, with $y>0$, are considered.


\subsection{Proof of Theorem \ref{thm:MainRes} - eGinOE}\label{subsec:ProoffiniteN}

\noindent
The mean self-overlap of eigenvectors in the eGinOE can be written as
\be\label{Eq:Overlaps}
\begin{split}
     \mathcal{O}^{(\eGinOE,c)}_{N}(z) &= \bigg\langle \frac{1}{N}\sum_{n=1}^N \mathcal{O}_{nn} \ \delta(z-z_n) \bigg\rangle_{\text{eGinOE},N} = \bigg\langle \mathcal{O}_{\widetilde{z}} \ \delta(z-\widetilde{z}) \bigg\rangle_{X} \ .
\end{split}
\ee
To evaluate this object we express $\mathcal{O}_z$ using Eqs. \eqref{overlapcomp2} and \eqref{prefacs} and then average this over the measure in Eq. \eqref{Eq:JPDF_eGinOE_IRSD}. As the first step we integrate over $dW = d\mathbf{w}_1 d\mathbf{w}_2$ by rewriting $\Tr WW^T = \mathbf{w}_1^T \mathbf{w}_1 + \mathbf{w}_2^T \mathbf{w}_2$. To do so we define the normalised average of any function $\mathcal{A}(\mathbf{w})$ with respect to $\mathbf{w}=\mathbf{w}_1$ or $\mathbf{w}=\mathbf{w}_2$ as
\be\label{Eq:WaverageDef}
    \bigg\langle \mathcal{A}(\mathbf{w}) \bigg\rangle_{\mathbf{w}} \equiv  \frac{(1-\tau^2)^{-\frac{N-2}{2}}}{(2\pi)^{\frac{N-2}{2}}} \int d\mathbf{w} \exp \left[ -\frac{1}{2(1-\tau^2)} \mathbf{w}^T \mathbf{w} \right] \mathcal{A}(\mathbf{w}) \ .
\ee
Introducing the constant $\widetilde{C}_{N,\tau} = C_{N,\tau}^\prime (2\pi)^{N-2} (1-\tau^2)^{N-2}$ then allows us to write
\begin{align}
    \bigg\langle \mathcal{O}_{\widetilde{z}} \ \delta(z-\widetilde{z}) \bigg\rangle_{X} &= \widetilde{C}_{N,\tau} \ e^{ - \frac{x^2}{1+\tau} - \frac{y^2}{1-\tau} } \ \int d\delta \ \frac{2y \delta}{\sqrt{\delta^2 + 4y^2}} \ e^{ -\frac{1}{2(1-\tau^2)} \delta^2  }    \int dX_2 \ e^{ -\frac{1}{2(1-\tau^2)}\Tr \left( X_2 X_2^T - \tau X_2^2 \right) } \nonumber \\
    &\times  \det \Big[ (x\eins_{N-2} - X_2)^2 + y^2 \eins_{N-2} \Big] \bigg\langle \bigg\langle \widetilde{c}_1 + \widetilde{c}_2 \   \left( \mathbf{b}_{N-2}^\dagger \mathbf{b}_{N-2} \right)    \bigg\rangle_{\mathbf{w}_2} \bigg\rangle_{\mathbf{w}_1}  \label{Eq:avOvleGinOEstep} \ .
\end{align}
We again adapt the statements from \cite{WCF}, only highlighting the necessary changes to Lemmas 3.1 and 3.2.

\begin{lem}\label{lem:Wav1}
Let $Y$ be an $N-2$ dimensional matrix and $\mathbf{w}_1$, $\mathbf{w}_2$ be two vectors, each of length $N-2$, with independent real entries. Then
\be\label{Eq:LemWaverage1}
\begin{split}
    &\bigg\langle \mathbf{w}^T Y \mathbf{w} \bigg\rangle_{\mathbf{w}} = (1-\tau^2) \Tr Y \quad \textup{and} \quad \bigg\langle \bigg\langle    \mathbf{w}_1^T Y \mathbf{w}_2  \bigg\rangle_{\mathbf{w}_2} \bigg\rangle_{\mathbf{w}_1} = \bigg\langle \bigg\langle     \mathbf{w}_2^T Y \mathbf{w}_1 \bigg\rangle_{\mathbf{w}_2} \bigg\rangle_{\mathbf{w}_1} = 0 \ ,
\end{split}
\ee
where $\langle \cdot \rangle_{\bm w}$ is defined in Eq. \eqref{Eq:WaverageDef} and $\mathbf{w} \in \{ \mathbf{w}_1, \mathbf{w}_2\}$.
\end{lem}
\begin{lem}\label{lem:Wav2}
    The average over $\mathbf{w}_1$ and $\mathbf{w}_2$ of the product $\mathbf{b}_{N-2}^\dagger \mathbf{b}_{N-2}$, defined in Eq. \eqref{Eq:bbdaggereGinOE}, is given by
    \be\label{Eq:LemWaverage2}
    \begin{split}
        \bigg\langle \bigg\langle \widetilde{c}_1 + \widetilde{c}_2 \   \left( \mathbf{b}_{N-2}^\dagger \mathbf{b}_{N-2} \right) \bigg\rangle_{\mathbf{w}_2} \bigg\rangle_{\mathbf{w}_1}
        &= \frac{1}{2} \left(2 + \frac{\delta^2}{2y^2}  \right) \bigg[ 1 +(1-\tau^2) \Tr \left[ B^{-1} \right] \bigg] \ .
    \end{split}
    \ee
\end{lem}

\noindent
Lemma \ref{lem:Wav2} allows us to perform the averages over $\bm w_1$ and $\bm w_2$ in Eq. \eqref{Eq:avOvleGinOEstep} yielding
\begin{align}
    \bigg\langle \mathcal{O}_{\widetilde{z}} \ \delta(z-\widetilde{z}) \bigg\rangle_{X} 
    &= \widetilde{C}_{N,\tau} e^{ - \frac{x^2}{1+\tau} - \frac{y^2}{1-\tau} } \ \int d\delta \ \frac{y \delta}{\sqrt{\delta^2 + 4y^2}} \ e^{ -\frac{1}{2(1-\tau^2)} \delta^2  } \  \left(2 + \frac{\delta^2}{2y^2}  \right)   \int dX_2 \ e^{-\frac{1}{2(1-\tau^2)}\Tr \left( X_2 X_2^T - \tau X_2^2 \right) } \nonumber \\ 
    &\times \det \left[ (x\eins_{N-2} - X_2)^2 +y^2 \eins_{N-2} \right] \bigg[ 1 + (1-\tau^2) \Tr \left[ B^{-1} \right] \bigg] \ . \label{Eq:avOvleGinOEstep2}
\end{align}
We now use the definition of the  matrix $B$ from Eq. (\ref{Bmat}1) and notice that
\be\label{detformulas}
\begin{split}
    \det B 
    &= \det \left[ (x\eins_{N-2} - X_2)^2 +y^2 \eins_{N-2} \right] = \det \left[ \begin{matrix}
    0 & i(z\eins_{N-2} - X_2)  \\
    i(\bar{z} \eins_{N-2} - X_2^T) & 0 \\
    \end{matrix}  \right] \ , 
\end{split}
\ee
which can be used to evaluate the first term of Eq. \eqref{Eq:avOvleGinOEstep2}. To evaluate the second term we once again appeal to Eq. \eqref{deriv} and so we introduce a block-matrix $M$ via
\be\label{Eq:propM1mat}
    M \equiv \left( \begin{matrix}
        \sqrt{\mu} \eins_{N} & i\left( z \eins_N - X \right) \\
        i\left( \bar{z} \eins_N - X^T \right) & \sqrt{\mu} \eins_N \\
    \end{matrix} \right) \ ,
\ee
where each block is of size $N\times N$ and $\mu$ is a real parameter. It is easy to see that with the replacement $N\rightarrow N-2$ and $\mu=0$ we have the same relations as in the case of the eGinUE, cf. Eq. \eqref{detMpartialdetM}. Furthermore, we define the normalised average of any function $\mathcal{A}(X_2)$ over $X_2$ by writing
\be\label{Eq:X2av}
      \bigg\langle \mathcal{A}(X_2) \bigg\rangle_{X_2} \equiv C_{N-2,\tau}^{-1} \int dX_2 \ \exp \left[-\frac{1}{2(1-\tau^2)}\Tr \left( X_2 X_2^T - \tau X_2^2 \right) \right] \ \mathcal{A}(X_2) \ ,
\ee
so that Eq. \eqref{Eq:avOvleGinOEstep2} becomes
\be\label{Eq:avOvleGinOEstep3}
\begin{split}
     \bigg\langle \mathcal{O}_{\widetilde{z}} \ \delta(z-\widetilde{z}) \bigg\rangle_{X} 
     &= \widehat{C}_{N,\tau} e^{ - \frac{1}{1+\tau} x^2  - \frac{1}{1-\tau} y^2 } \ \int d\delta \ \frac{y \delta}{\sqrt{\delta^2 + 4y^2}} \ e^{ -\frac{1}{2(1-\tau^2)} \delta^2  } \ \left(2 + \frac{\delta^2}{2y^2}  \right)  \\ 
     &\times \bigg[ \  \bigg\langle \det M \bigg\rangle_{X_2} \ \bigg\vert_{\mu = 0} + (1-\tau^2) \frac{\partial}{\partial \mu }\bigg\langle \det M \bigg\rangle_{X_2} \ \bigg\vert_{\mu = 0} \ \bigg] \ ,
\end{split}
\ee
with the new constant $\widehat{C}_{N,\tau} \equiv \widetilde{C}_{N,\tau} C_{N-2,\tau}$. In evaluating the averages over matrices $X_2$ we rely on the following

\begin{prop}\label{prop:X2av}
     Let $M$ be the matrix defined in Eq. \eqref{Eq:propM1mat}, where $z$ is a complex number and $X$ is an $N\times N$ eGinOE matrix. Then
    \begin{align}
        \bigg\langle \det M \bigg\rangle_{X} &= \frac{1}{4\pi^2} \int_{\mathbb{C}} du d\bar{u} \ \int_{\mathbb{R}^2} dv_1 \, dv_2 \ e^{  -\vert u \vert^2 -\frac{v_1^2   + v_2^2 }{2} } \int_0^{\infty} dR e^{-R} \Bigg[  \big(R + g(z,u,v_1,v_2) \big)^2  +  \mu^2   \label{Eq:detMreseGinOEb}   \\
        & + 2\mu \big( g(z,u,v_1,v_2) - R \big)  \Bigg]^{\frac{N}{2}} P_N\left(\frac{\mu + R + g(z,u,v_1,v_2)}{\sqrt{\big(R + g(z,u,v_1,v_2) \big)^2 + \mu^2 +2\mu \big(g(z,u,v_1,v_2) - R\big)}}\right) \ ,  \nonumber
    \end{align}
    where $P_N(t) $ is a Legendre polynomial defined in Eq. \eqref{LegPolyint} and $g(z,u,v_1,v_2) = \tau \ \vert u \vert^2  + \left( \bar{z} -i v_2 \sqrt{\tau} \right) \left( z +i v_1 \sqrt{\tau} \right)$.
\end{prop}

\begin{proof}
    We use Eq. \eqref{detIdentGrass} to express the determinant as an integral over anticommuting variables as 
    \be
        \bigg\langle \det M \bigg\rangle_{X} = (-1)^N\expval{ \int D \bm \Phi D \bm \chi \exp{ -  \left( \begin{matrix} \bm \phi_1^T & \bm \phi_2^T \end{matrix} \right) \left( \begin{matrix} \sqrt{\mu} \eins_N & i \left( z \eins_N - X \right) \\ i \left( \bar{z} \eins_N - X^T \right) & \sqrt{\mu} \eins_N 
        \end{matrix} \right) \left( \begin{matrix} \bm \chi_1 \\ \bm \chi_2 \end{matrix} \right)
        }  }_{X} \ ,
    \ee
    which upon using $\bm \phi_i^T X \bm \chi_j = - \Tr \left[ X \bm \chi_j \bm \phi_i^T \right]$ and writing both traces in terms of $X$ yields
    \begin{align}
        \bigg\langle \det M \bigg\rangle_{X} 
        &= (-1)^N  \int D \bm \Phi D \bm \chi \, e^{- \sqrt{\mu} \left( \bm \phi_1^T \bm \chi_1 + \bm \phi_2^T \bm \chi_2\right)  - i \left( z \bm \phi_1^T \bm \chi_2 + \bar{z} \bm \phi_2^T \bm \chi_1\right)   } \bigg\langle 
        e^{- i \Tr \left[ X  (\bm \chi_2 \bm \phi_1^T - \bm \phi_2 \bm \chi_1^T )  \right]} \bigg\rangle_{X}  \ .
        \label{eq:detM_eGinOE_step1}
    \end{align}
    Defining $A =  i \left( \bm \chi_2 \bm \phi_1^T - \bm \phi_2 \bm \chi_1^T \right)$ we then proceed by utilising the following identity \cite{FT}
    \begin{equation}
        \expval{e^{-\Tr[X A] }}_{X} = \exp{ \frac{1}{2} \Tr[ A A^T + \tau A^2]  } \ ,
        \label{eq:ave_exp_Tr_GNell}
    \end{equation}
    which yields
    \be\label{EqXell}
    \begin{split}
        \expval{e^{-\Tr[X A] }}_{X} &= \exp \bigg[ \left( \bm \phi_1^T \bm \chi_1 \right) \left( \bm \phi_2^T \bm \chi_2 \right) + \frac{\tau}{2}\left( \bm \phi_1^T \bm \chi_2 \right)^2  + \frac{\tau}{2} \left( \bm \chi_1^T \bm \phi_2 \right)^2   + \tau \left( \bm \phi_1^T \bm \phi_2 \right) \left( \bm \chi_2^T  \bm \chi_1 \right) \bigg] \ .
    \end{split}
    \ee
    The exponential of the terms quartic in Grassmann variables can be expressed using Hubbard-Stratonovich transformations as in Eq. \eqref{ellHubbStratauxCOMPLEX} and the additional relationship 
    \be\label{ellHubbStrataux}
    \begin{split}
        \exp \bigg[ \tau \left( \bm \phi_1^T \bm \phi_2  \right) \left( \bm \chi_2^T  \bm \chi_1 \right) \bigg] &= \frac{1}{2\pi} \int_{\mathbb{C}} du d\bar{u} \ \exp \bigg[ -\vert u \vert^2 - u \sqrt{\tau} \left( \bm \chi_2^T  \bm \chi_1 \right) - \bar{u}\sqrt{\tau} \left( \bm \phi_1^T \bm \phi_2  \right) \bigg] \ .
    \end{split}
    \ee
    These transformations can be used to re-express the integrals in Eq. \eqref{eq:detM_eGinOE_step1} in terms of the matrix
    \be\label{MPfaffMat2}
        H =  \left(\begin{matrix}
        0 &  \bar{u} \sqrt{\tau} &  \sqrt{\mu} + \bar{q}  &  i z - v_1 \sqrt{\tau}  \\
        -  \bar{u} \sqrt{\tau}  & 0 &  i \bar{z} + v_2 \sqrt{\tau} & \sqrt{\mu} + q  \\
        -\left( \sqrt{\mu} + \bar{q} \right)  & -\left( i \bar{z} + v_2 \sqrt{\tau} \right) & 0 & - u \sqrt{\tau} \\
        -\left( i z - v_1 \sqrt{\tau} \right) & -\left( \sqrt{\mu} + q \right) &  u \sqrt{\tau} & 0 \\
        \end{matrix} \right) \otimes \eins_N = F \otimes \eins_N \ ,
    \ee
    therefore, as follows from Eq. \eqref{PfaffIdentGrass}, one must now calculate $\Pf(H)$. The Pfaffian obeys the following rule 
    \be
        \Pf \left( A \otimes B \right) = (-1)^{nm(m-1)/2} \left( \Pf A\right)^{m} \left( \det B \right)^{n}
    \ee
    for matrices $A$ of dimension $2 n$ and symmetric matrices $B=B^T$ of dimension $m$ \cite[Eq. (29)]{DVarjas}. In our case we have $A=F$ and $B=\eins_N$, meaning that $m=N$ and $n=2$, this results in $\Pf H = \left( \Pf F\right)^{N}$. The Pfaffian of a $4\times 4$ anti-symmetric matrix can be easily evaluated as
    \be
        \Pf \left[ \begin{matrix}
        0 &  a &  b  &  c  \\
        -a  & 0 &  d & e  \\
        -b  & -d & 0 & f \\
        -c & -e & -f & 0 \\
        \end{matrix} \right] = af - be +dc \ ,
    \ee
    which in our case implies that
    \be\label{MPfaffMat4}
        \Pf H = (-1)^N \bigg[ \left( \sqrt{\mu} + \bar{q} \right) \left( \sqrt{\mu} + q \right)  +\tau \ \vert u \vert^2   + \left( \bar{z} -i v_2 \sqrt{\tau} \right) \left( z +i v_1 \sqrt{\tau} \right)  \bigg]^N
    \ee
    and overall we find
    \be\label{MeEllres1}
    \begin{split}
        \bigg\langle \det M \bigg\rangle_{X}&= \frac{1}{8\pi^3} \int_{\mathbb{C}} dq d\bar{q}  \int_{\mathbb{C}} du d\bar{u} \int_{\mathbb{R}^2} dv_1 dv_2 \, e^ {-\vert q \vert^2 -\vert u \vert^2 -\frac{v_1^2   + v_2^2}{2} } \bigg[ \left( \sqrt{\mu} + \bar{q} \right) \left( \sqrt{\mu} + q \right)  + g(z, u, v_1, v_2)  \bigg]^N \, .
    \end{split}
    \ee
    This finally yields Eq. \eqref{Eq:detMreseGinOEb} after changing $q = \sqrt{R} e^{i\theta}$ and employing the definition of a Legendre polynomial, given in Eq. \eqref{LegPolyint}. 
\end{proof}

\noindent
We now state the required results involving averages over $X_2$ in the following Corollary:

\begin{cor}\label{cor:AvdetX2}
    With the average taken as in Eq. \eqref{Eq:X2av}, using the $N-2$ sized eGinOE matrix $X_2$ and the matrix $M$ given by Eq. \eqref{Eq:propM1mat}, we have
    \be\label{Eq:AvdetMemu0}
        \bigg\langle \det M \bigg\rangle_{X_2} \ \bigg\vert_{\mu=0} = (N-2)! \, P_{N-2}
    \ee
    and
    \be\label{Eq:AvdetMepartialmu0}
        \frac{\partial}{\partial \mu} \bigg\langle \det M \bigg\rangle_{X_2} \ \bigg\vert_{\mu=0} = (N-2)! \bigg[ P_{N-3} + (N-3) \, P_{N-4} -  T_{N-4} \bigg]  
    \ee
    with $P_N$ and $T_N$ defined in Eqs. \eqref{Eq:PNRes} and \eqref{Eq:TNRes} respectively. 
\end{cor}

\begin{proof}
    The proof follows similar ideas as in the eGinUE case and in \cite{WCF}. To show Eq. \eqref{Eq:AvdetMemu0} we directly use Proposition \ref{prop:X2av} for $N-2$ and the Legendre polynomial property $P_{N-2}(1)=1$, thus Eq. \eqref{Eq:detMreseGinOEb} becomes
    \be\label{Eq:AvdetMequiv}
    \begin{split}
        \bigg\langle \det M \bigg\rangle_{X_2} \ \bigg\vert_{\mu=0} &=\frac{1}{4\pi^2} \int_{\mathbb{C}} du d\bar{u} \ \int_{\mathbb{R}^2} dv_1 \ dv_2 \ e^{-\vert u \vert^2 -\frac{1}{2} \left( v_1^2   + v_2^2 \right) } \int_0^{\infty} dR  \ e^{-R}   \  \Big(R + g(z,u,v_1,v_2) \Big)^{N-2} \, .
    \end{split}
    \ee
    To evaluate this, we employ polar coordinates $u' = \sqrt{R'} e^{i\theta'}$, use the binomial theorem twice, which allows us to perform the integrals over $R$ and $R'$, and finally we invoke the integral representation of the Hermite polynomials given in Eq. \eqref{Eq:HermitePoly}. After these manipulations we find that
    \begin{align}
        \bigg\langle \det M \bigg\rangle_{X_2} \ \bigg\vert_{\mu=0} &= (N-2)! \sum_{k=0}^{N-2} \tau^k \sum_{m=0}^k \frac{1}{m!} \HE_m\left( \frac{z}{\sqrt{\tau}} \right) \HE_m\left( \frac{\bar{z}}{\sqrt{\tau}} \right) =
        (N-2)! \, P_{N-2} \ ,
        \label{eq:detM_eGinOE_eval}
    \end{align}
    where the final equality made use of the Christoffel-Darboux formula
    \be\label{CDform}
        \sum_{m=0}^{k} \frac{1}{m!} \HE_m(x) \HE_m(y) = \frac{1}{k!} \frac{\HE_k(y) \HE_{k+1}(x) - \HE_k(x) \HE_{k+1}(y)}{x-y} \ .
    \ee
    To prove Eq. \eqref{Eq:AvdetMepartialmu0}, we start with Eq. \eqref{Eq:detMreseGinOEb}, replace $N\rightarrow N-2$, differentiate over $\mu$ using Eq. \eqref{diffLegend} and simplify using that $P_{N-2}^\prime (1)= (N-1) (N-2)/2$. Collecting all terms we find that
    \be
    \begin{split}
        \frac{\partial}{\partial \mu} \bigg\langle \det M \bigg\rangle_{X_2} \ \bigg\vert_{\mu=0} &= \frac{N-2}{4\pi^2} \int_{\mathbb{C}} du d\bar{u} \ \int_{\mathbb{R}^2} dv_1 \, dv_2 \ e^{-\vert u \vert^2 -\frac{1}{2} \left( v_1^2   + v_2^2 \right) } \\
        &\times \int_0^{\infty} dR  \ e^{-R}   \, \bigg[ \Big(R + g(z,u,v_1,v_2) \Big)^{N-3}  + (N-3) \, R \, \Big(R + g(z,u,v_1,v_2) \Big)^{N-4}  \bigg] 
    \end{split}
    \ee
    and from Eqs. \eqref{Eq:AvdetMemu0} and \eqref{Eq:AvdetMequiv} immediately identify the first of the two terms to be $(N-2)!P_{N-3}$. The second term can be evaluated using the same procedure that was used to arrive at Eq. \eqref{eq:detM_eGinOE_eval}, this produces
    \be
    \begin{split}
        &\frac{(N-2)(N-3)}{4\pi^2} \int_{\mathbb{C}} du d\bar{u} \ \int_{\mathbb{R}^2} dv_1  dv_2 \ e^{- \vert u \vert^2 - \frac{1}{2} \left( v_1^2   + v_2^2 \right)} \int_0^\infty dR \ R \ e^{-R}   \Big( R + g(z,u,v_1,v_2) \Big)^{N-4} \\
        &= (N-2)! \sum_{k=0}^{N-4} \ \frac{\tau^k}{k!} \ (N-3-k) \sum_{m=0}^k \HE_m\left( \frac{z}{\sqrt{\tau}} \right)  \ \HE_m\left( \frac{\bar{z}}{\sqrt{\tau}} \right) =(N-2)!\Big( (N-3)P_{N-4}-T_{N-4} \Big) \ .
    \end{split}
    \ee
    Collecting both terms we arrive at Eq. \eqref{Eq:AvdetMepartialmu0} and conclude the proof of this Corollary.
\end{proof}

\noindent
We can now finish the proof of Theorem \ref{thm:MainRes} by returning to Eq. \eqref{Eq:avOvleGinOEstep3} and inserting the results from Corollary \ref{cor:AvdetX2}. The final remaining integral over $\delta$ can be easily computed as
\be\label{deltaInt}
    \frac{1}{2y} \ \int_0^\infty d\delta \ \delta \  \sqrt{\delta^2 + 4y^2} \ e^{-\frac{1}{2(1-\tau^2)}\delta^2} = (1-\tau^2) + \sqrt{\frac{\pi}{2}} \ e^{\frac{2y^2}{1-\tau^2}} \ \left(1-\tau^2 \right)^{\frac{3}{2}} \ \frac{1}{2\vert y \vert } \ \text{erfc}\left(\sqrt{\frac{2}{1-\tau^2}} \ \vert y \vert \right)
\ee
and, after simple manipulations, the overall constant can be evaluated. Thus we find
\begin{align}
   \mathcal{O}^{(\eGinOE,c)}_{N}(z) &=  \frac{1}{\pi} \sqrt{\frac{1-\tau}{1+\tau}} \ \exp \left[ -\frac{x^2}{1+\tau} -\frac{y^2}{1-\tau} \right] \bigg[ 1 + \sqrt{\frac{\pi(1-\tau^2)}{2}} \ \exp \left[ \frac{2 y^2}{1-\tau^2} \right] \ \frac{1}{2\vert y \vert} \ \textup{erfc}\left( \sqrt{\frac{2 }{1-\tau^2}} \ \vert y \vert \right) \bigg] \nonumber \\ 
    &\times \bigg[   P_{N-2} + (1-\tau^2) \bigg(  P_{N-3} + (N-3) P_{N-4} - T_{N-4} \bigg) \bigg] \ ,\label{Oexact}
\end{align}
which concludes the proof of Theorem \ref{thm:MainRes}.


\section{Asymptotic Analysis for large matrix size}\label{sec:AsymptoticAnalysis}


\subsection{Preliminaries for the asymptotic analysis}\label{subsubsec:AsymptoticAnalysisPreliminaries}

In this subsection, we begin by outlining key notation for our asymptotic analysis and then give a set of Lemmas containing useful integral representations of the building blocks $\rho_{N}^{\text{eGinUE,c}}(z)$ and $R_N$ in the eGinUE and $P_N$ and $T_N$ in the eGinOE. We begin by recalling the incomplete $\Gamma$-function defined as
\be\label{Eq:incompleteGamma}
    \Gamma\left( N , a \right) = \Gamma\left(N\right) \, e^{-a} \ \sum_{k=0}^{N-1} \frac{a^k}{k!} =  \int_{a}^{\infty} du \ u^{N-1} \ e^{-u} 
\ee
and then we further define the function
\be\label{Thetafunction}
    \Theta_N^{(M)}(x) \equiv \frac{\Gamma(N-M+1,Nx)}{\Gamma(N-M+1)}\ .
\ee
Note that for a fixed non-negative integer $M$ the function $\Theta_N^{(M)}(x)$ tends towards the Heaviside step function $\Theta [1-x]$ as $N \to \infty$.  For further details of the above statements, see \cite{WCF}. We can also state the following:

\begin{lem}\label{lem:rhoNintegralrep}
    The density, $\rho^{(\textup{eGinUE},c)}_N(z)$, given in Eq. \eqref{eq:eGinUE_density} via Hermite polynomials has the integral representation
    \be\label{RNdeffin3}
    \begin{split}
        \rho^{(\textup{eGinUE},c)}_N(z) &= \frac{N}{2\pi^2 \tau} \frac{1}{\sqrt{1-\tau^2}}\exp \left[ -\frac{ \vert z \vert^2 - \tau \ \textup{Re}(z^2) }{1-\tau^2} \right]  \exp \left[ \frac{z^2 +\bar{z}^2}{2\tau} \right]  \int_{-\infty}^{\infty} dp  \int_{-\infty}^{\infty} dq \ \Theta_N^{(1)} \left( \frac{p^2 - q^2}{2} \right) \\
        &\times \exp \left[ \frac{N}{2} \left( p^2 - q^2 \right) -\frac{N}{2\tau }\left( p^2 + q^2 \right) - i \frac{\sqrt{2N}\RE(z)}{\tau} \, q  - \frac{\sqrt{2N}\IM(z)}{\tau} \, p  \right]  \ . 
    \end{split}
    \ee
\end{lem}

\begin{lem}\label{lem:RNpintegralrep}
    The quantity $R_N$, given in Eq. \eqref{Eq:RNRes}, via Hermite polynomials has the integral representation
    \begin{align}
        R_N &= \frac{N^2}{2\pi^2 \tau} \frac{1}{\sqrt{1-\tau^2}}\exp \left[ -\frac{\vert z \vert^2 - \tau \ \textup{Re}(z^2)}{1-\tau^2} \right] \exp \left[ \frac{z^2 +\bar{z}^2}{2\tau} \right] \int_{-\infty}^{\infty} dp  \int_{-\infty}^{\infty} dq \ \Theta_N^{(1)}\left(\frac{p^2 - q^2}{2} \right) \ \frac{N}{2} \left( p^2 - q^2 \right) \nonumber \\
        &\times \exp \left[ \frac{N}{2} \left( p^2 - q^2 \right) - \frac{N}{2\tau }\left( p^2 + q^2 \right) - i \frac{\sqrt{2N}\RE(z)}{\tau} \, q  - \frac{\sqrt{2N}\IM(z)}{\tau} \, p  \right] \ . \label{RNpdeffin3}
    \end{align}
\end{lem}

\begin{lem}\label{lem:PNintegralrep}
    The quantity $P_N$, given in Eq. \eqref{Eq:PNRes}, via Hermite polynomials has the integral representation
    \be\label{PNdeffin3}
    \begin{split}
        P_N&= N \, \sqrt{2N} \, \frac{1}{2\IM(z)} \frac{(-1)}{2\pi \tau} \exp \left[ \frac{\RE(z)^2 -\IM(z)^2}{\tau} \right] \int_{-\infty}^{\infty} dp \ p  \int_{-\infty}^{\infty} dq \ \Theta_N^{(0)}\left(\frac{ p^2 - q^2 }{2}\right)  \\
        &\times \exp \left[ \frac{N}{2} \left( p^2 - q^2 \right) -\frac{N}{2\tau }\left( p^2 + q^2 \right) - i \frac{\sqrt{2N}\RE(z)}{\tau} \, q  - \frac{\sqrt{2N}\IM(z)}{\tau} \, p  \right] \ .
    \end{split}
    \ee
\end{lem}

\begin{lem}\label{lem:TNintegralrep}
    The quantity $T_N$ given in Eq. \eqref{Eq:TNRes} via Hermite polynomials has the integral representation
    \be\label{TNdef3}
    \begin{split}
         T_N &= N \, \sqrt{2N} \, \frac{1}{2\IM(z)} \frac{(-1)}{2\pi \tau} \exp \left[ \frac{\RE(z)^2 -\IM(z)^2}{\tau} \right] \int_{-\infty}^{\infty} dp \ p  \int_{-\infty}^{\infty} dq \ \Theta_N^{(1)}\left(\frac{ p^2 - q^2 }{2}\right) \ \frac{N}{2} \left( p^2 - q^2 \right) \\
        &\times \exp \left[ \frac{N}{2} \left( p^2 - q^2 \right) -\frac{N}{2\tau }\left( p^2 + q^2 \right) - i \frac{\sqrt{2N}\RE(z)}{\tau} \, q  - \frac{\sqrt{2N}\IM(z)}{\tau} \, p  \right] \ .
    \end{split}
    \ee
\end{lem}

\noindent
The verification of the above Lemmas is left to the reader as it a is fairly straightforward exercise in employing the following forms of the Hermite polynomials:
\begin{equation}
    \HE_k\left( \frac{\bar{z}}{\sqrt{\tau}} \right) = \frac{i^k e^{\frac{\bar{z}^2}{2 \tau}}}{\sqrt{2 \pi \tau}}  \int_{-\infty}^\infty dt_1 e^{ -\frac{t_1^2}{2 \tau} - i \frac{\bar{z}}{\tau} t_1 } \left(\frac{t_1}{\sqrt{\tau}}\right)^k   , \	\HE_k\left( \frac{z}{\sqrt{\tau}} \right) = \frac{(-i)^k e^{\frac{z^2}{2 \tau}}}{\sqrt{2 \pi \tau}}  \int_{-\infty}^\infty dt_2 e^{ -\frac{t_2^2}{2 \tau} + i \frac{z}{\tau} t_2 } \left(\frac{t_2}{\sqrt{\tau}}\right)^k  \ .
	\label{eq:He_k_int}
\end{equation}


\subsection{Asymptotic Analysis in strong non-Hermiticity limit}\label{subsubsec:AsymptoticAnalysisSTRONG}

We now proceed to prove Corollaries \ref{cor:eGinOEbulkStrong} and \ref{cor:eGinOEdepletionStrong} by evaluating the large $N$ limit of the two relevant integral representations given in the previous section, i.e. Lemmas \ref{lem:PNintegralrep} and \ref{lem:TNintegralrep}. We present the limiting expressions of $P_N$ and $T_N$, before proving the main statements in Section \ref{subsubsec:ProofAsymptbulk} and Section \ref{subsubsec:ProofAsymptdepletion}.


\subsubsection{eGinOE: Large $N$ expansion of $P_N$ and $T_N$}\label{subsubsec:PNTNAsymptotic}

In this section we evaluate the large $N$ limit of the quantities $P_N$ and $T_N$. To facilitate this evaluation it is convenient to combine the integral representations of $P_N$ and $T_N$ into a single formula
\begin{equation}
    \begin{split}
        A_{N}\big(M,F(p,q)\big) =& N^{M+1} \, \sqrt{2N}  \frac{1}{2\IM(z)} \frac{(-1)}{2\pi \tau} \exp \left[ \frac{\RE(z)^2 -\IM(z)^2}{\tau} \right] \int_{-\infty}^{\infty} dp \ p  \int_{-\infty}^{\infty} dq \ \Theta_N^{(M)}\left(\frac{p^2 - q^2}{2}\right)  \\
        &\times F(p,q) \times \exp \left[ \frac{N}{2} \left( p^2 - q^2 \right) - \frac{N}{2\tau }\left( p^2 + q^2 \right) - i \frac{\sqrt{2N}\RE(z)}{\tau} \, q  - \frac{\sqrt{2N}\IM(z)}{\tau} \, p  \right]   \ ,
    \end{split}
\end{equation}
where choosing $M=0$ and $F(p,q) = 1$ reproduces $P_N$, while choosing $M=1$ and $F(p,q) = (p^2 - q^2)/2$ yields $T_N$. Rearranging the terms in the exponential we then write
\begin{align}
    A_{N}\big(M,F(p,q)\big) =& N^{M+1} \,  \frac{\sqrt{2N}}{2\IM(z)} \frac{(-1)}{2\pi \tau} \exp \left[ \frac{\RE(z)^2 -\IM(z)^2}{\tau} \right] \int_{-\infty}^{\infty} dp \ p  \ \exp[ - N \left(\frac{p^2(1- \tau)}{\tau} + \frac{\sqrt{2} \IM(z) p}{\tau \sqrt{N}} \right)]  \nonumber \\
    & \times \int_{-\infty}^{\infty} dq \ \Theta_N^{(M)}\left(\frac{p^2 - q^2}{2}\right) F(p,q) \exp[ - N \left(\frac{q^2(1 + \tau)}{\tau} + \frac{i \sqrt{2} \RE(z) q}{\tau \sqrt{N}} \right)] \ ,
\end{align}
and evaluate the integrals using Laplace's method, i.e. for ${N \gg 1}$ replace
\begin{equation}
     \int_a^b dx g(x) e^{-N \mathcal{L}(x)} \approx \sqrt{ \frac{2 \pi}{N \mathcal{L}''(x_*)} } g(x_*) e^{-N\mathcal{L}(x_*)} \Theta[x_*-a] \Theta[b - x_*] \ ,
\end{equation}
where $x_*= \text{argmin} \,\mathcal{L}(x)$, assuming that  $g(x_*)\ne 0$. It can be easily seen that the $p$ and $q$ integrals are dominated by the vicinity of
\begin{equation}
    p^* = - \frac{\sqrt{2} \IM(z)}{1 - \tau} \frac{1}{\sqrt{N}} \hspace{1cm}  \text{and} \hspace{1cm} q^* = - \frac{i\sqrt{2} \RE(z)}{1 + \tau} \frac{1}{\sqrt{N}} \ ,
\end{equation}
which straightforwardly yields
\begin{equation}
\begin{split}
    A_{N\gg 1}\big(M,F(p,q)\big) \approx& \frac{N^{M}F(p^*, q^*)}{(1-\tau) \sqrt{1 - \tau^2}}  \Theta_N^{(M)} \left[ \frac{1}{N} \left( \frac{\RE(z)^2}{(1 + \tau)^2} + \frac{\IM(z)^2}{(1 - \tau)^2} \right) \right] \exp[ \frac{\RE(z)^2}{1 + \tau} + \frac{\IM(z)^2}{1 - \tau} ] \ .
\end{split}
\end{equation}
Specializing for $P_N$ and $T_N$, one then finds that
\begin{align}
    P_{N\gg 1} =& A_{N\gg 1}\big(M = 0, F(p,q) = 1) \label{PN_int_asymptotic_form} \\
    \approx& \frac{1}{(1-\tau) \sqrt{1 - \tau^2}}  \Theta_N^{(0)} \left[ \frac{1}{N} \left( \frac{\RE(z)^2}{(1 + \tau)^2} + \frac{\IM(z)^2}{(1 - \tau)^2} \right) \right] \exp[ \frac{\RE(z)^2}{1 + \tau} + \frac{\IM(z)^2}{1 - \tau} ] \nonumber
\end{align}
and
\begin{align}
    T_{N\gg 1} =& A_N\left(M = 1, F(p,q) = \frac{p^2 - q^2}{2} \right) \label{TN_int_asymptotic_form} \\
    \approx& \frac{1}{\left( 1-\tau \right) \sqrt{1-\tau^2}}   \ \bigg[ \frac{ \RE(z)^2}{(1+\tau)^2} + \frac{ \IM(z)^2}{(1-\tau)^2} \bigg]  \Theta_{N}^{(1)} \left[ \frac{1}{N}\left( \frac{\RE(z)^2}{(1+\tau)^2}  +  \frac{\IM(z)^2}{(1-\tau)^2} \right) \right] \exp \left[ \frac{\RE(z)^2}{1+\tau } + \frac{\IM(z)^2}{1-\tau}  \right]  \nonumber \ .
\end{align}


\subsubsection{Proof in the bulk of the droplet}\label{subsubsec:ProofAsymptbulk}

\noindent
We can now proceed to prove Corollary \ref{cor:eGinOEbulkStrong}. 

\begin{proof}
    The bulk scaling limit is obtained when $z=\sqrt{N}w$ with $w=x+iy$, $|w|<1$ and $|y| \gg N^{-1/2}$. Firstly, since $\IM(z)\approx O(\sqrt{N})$ we use the large argument behaviour of $\erfc(x)$, i.e. $\text{erfc}(x) \approx e^{-x^2}/(\sqrt{\pi} x)$ for $x \gg 1$, to see that the term containing the error function in Eq. \eqref{Eq:MainRes} vanishes as $N \to \infty$. Furthermore, we employ the large $N$ properties of $\Theta_N^{(M)}(s)$ to ascertain the leading order contributions to  Eqs. \eqref{PN_int_asymptotic_form} and \eqref{TN_int_asymptotic_form} in this limit, which are given by
    \begin{equation}
        P_{N\gg 1} \approx \frac{1}{(1-\tau) \sqrt{1 - \tau^2}} \exp[ \frac{Nx^2}{1+\tau} + \frac{Ny^2}{1-\tau}] \Theta\left[ 1 - \frac{x^2}{(1+\tau)^2} - \frac{y^2}{(1-\tau)^2} \right]
    \end{equation}
    and
    \begin{equation}
        T_{N\gg 1} \approx \frac{1}{(1-\tau) \sqrt{1 - \tau^2}} \exp[ \frac{Nx^2}{1+\tau} + \frac{Ny^2}{1-\tau}] \left( \frac{Nx^2}{(1+\tau)^2} + \frac{Ny^2}{(1-\tau)^2} \right) \Theta\left[ 1 - \frac{x^2}{(1+\tau)^2} -  \frac{y^2}{(1-\tau)^2} \right]
    \end{equation}
    respectively. Recalling the finite $N$ form of the mean self-overlap in the eGinOE, Eq. \eqref{Eq:MainRes}, one can see that in the large $N$ limit the dominant terms will be $NP_N$ and $T_N$. After some manipulation this yields
    \begin{equation}
       \lim_{N\to \infty} \frac{\mathcal{O}_N}{N} =\frac{1}{\pi} \left( 1 - \frac{x^2}{(1+\tau)^2} -  \frac{y^2}{(1-\tau)^2} \right) \Theta\left[ 1 - \frac{x^2}{(1+\tau)^2} -  \frac{y^2}{(1-\tau)^2} \right] \ ,
    \end{equation}
    immediately implying Eq. \eqref{Eq:eGinOEbulkStrongRes}.
\end{proof}


\subsubsection{Proof in the depletion regime of the droplet}\label{subsubsec:ProofAsymptdepletion}

\noindent
We now turn to the depletion region of the droplet in the eGinOE and the proof of Corollary \ref{cor:eGinOEdepletionStrong}.

\begin{proof}
    The mean self-overlap in the depletion regime of the eGinOE can be found by scaling $z=\sqrt{N} \delta + i\xi$ and taking $N\rightarrow \infty$, with $\delta,\xi \sim O(1)$. Considering the finite $N$ equation for the mean self-overlap, Eq. \eqref{Eq:MainRes}, we must therefore understand how the functions $P_N$ and $T_N$ behave in this parameter region. This can be done using Eqs. \eqref{PN_int_asymptotic_form} and \eqref{TN_int_asymptotic_form}, where one can obtain the leading-order asymptotic behaviour of $P_N$ and $T_N$ as
    \be
    \begin{split}
        P_{N\gg 1} & \approx  \frac{1}{(1-\tau) \sqrt{1-\tau^2}}  \ \Theta_N^{(0)} \left( \left[ \frac{\delta^2}{(1+\tau)^2} + \frac{\xi^2}{N(1-\tau)^2}  \right] \right) \  \exp \left[     \frac{N \delta^2}{1+\tau } + \frac{\xi^2}{1-\tau}  \right]
    \end{split} \label{eq:P_N_dep}
    \ee
    and
    \be
    \begin{split}
        T_{N\gg 1} &\approx \frac{1}{\left( 1-\tau \right) \sqrt{1-\tau^2}}   \ \bigg[ \frac{ N \delta^2}{(1+\tau)^2} + \frac{ \xi^2}{(1-\tau)^2} \bigg] \ \Theta_{N}^{(1)} \left[ \frac{\delta^2}{(1+\tau)^2} + \frac{\xi^2}{N(1-\tau)^2}  \right] \  \exp \left[ \frac{N \delta^2}{1+\tau } + \frac{\xi^2}{1-\tau}  \right]\\
    \end{split} \label{eq:T_N_dep}
    \ee 
    respectively. In fact, in the limit of large $N$  the same asymptotic behaviour  describes all $P_{N-m}$ and $T_{N-m}$, where $m$ is some fixed integer. This implies that for large $N$
    \begin{equation} 
        \begin{split}
            P_{N-2} + (1 - \tau^2) &\Big( P_{N-3} + (N-3) P_{N-4} - T_{N-2} \Big) \approx \\
            & N \ \frac{\sqrt{1 - \tau^2}}{ 1 - \tau } \exp \left[ \frac{N \delta^2}{1+\tau} + \frac{\xi^2}{1-\tau} \right] \left( 1 - \frac{\delta^2}{(1 + \tau)^2} \right) \Theta\left[ 1 - \frac{\delta^2}{(1 + \tau)^2} \right] \ .
        \end{split}
    \end{equation}
    Substituting the above into Eq. \eqref{Eq:MainRes} and manipulating the prefactors  one can now see that in such a limit
    \begin{equation}
       \lim_{N\to \infty} \frac{1}{N} \mathcal{O}^{\text{(eGinOE,c)}}_{N} (z) = \frac{1}{\pi} \left[  1 + \sqrt{\frac{\pi(1-\tau^2)}{2}} \  \frac{e^{ \frac{2 \xi^2}{1-\tau^2}} }{2\vert \xi \vert} \textup{erfc}\left( \sqrt{\frac{2 }{1-\tau^2}} \ \vert \xi \vert \right) \right] \left( 1 - \frac{\delta^2}{(1 + \tau)^2} \right) \Theta\left[ 1- \frac{\delta^2}{(1 + \tau)^2} \right] 
    \end{equation}
    and Eq. \eqref{Eq:eGinOEdepletionStrongRes} immediately follows.
\end{proof}

\noindent
\begin{rem}
  The density of complex eigenvalues in the depletion regime of the eGinOE, given in Eq. \eqref{Eq:density_depletion_strip_strong}, can be similarly derived from the knowledge of the asymptotic behaviour of $P_N$ for $N\gg 1$. This is done by using $z = \sqrt{N} \delta + i \xi$ and inserting the asymptotic form of $P_N$, Eq. \eqref{eq:P_N_dep}, into the finite $N$ equation for the density, Eq. \eqref{eq:eGinUE_density}, then taking the limit $N \to \infty$.
\end{rem}


\subsection{Asymptotic Analysis in weak non-Hermiticity limit}\label{subsubsec:AsymptoticAnalysisWEAK}

In this subsection we provide details of our asymptotic analysis of the mean density of complex eigenvalues and the mean self-overlap of eigenvectors at weak non-Hermiticity, in both the eGinUE and eGinOE. 


\subsubsection{eGinUE: Proof of Corollary \ref{cor:eGinUEbulkWeak}}\label{subsubsec:ProofeGinUEbulkWeak}

\begin{proof}
    The expression in Eq. \eqref{Eq:MainResCOMPLEX} implies that in order to evaluate the mean self-overlap in the WNH limit we must first evaluate $\rho_N^\text{{(eGinUE,c)}}(z)$ and the quantity $R_N$ in this limit, where
    \be\label{eq:WNH_tau_z}
        \tau = 1-\frac{\pi^2\alpha^2}{2N} \quad \text{and} \quad \RE(z) = \sqrt{N}X,  \quad \IM(z) = \frac{\pi y}{\sqrt{N}} \ .
    \ee
    Lemma \ref{lem:rhoNintegralrep} gives us a suitable integral representation for the density of complex eigenvalues
    \begin{align} 
        \rho^{(\textup{eGinUE},c)}_N(z) &=  \frac{N}{2 \pi^2 \tau} \frac{1}{\sqrt{1-\tau^2}}\exp \left[ -\frac{ \vert z \vert^2 - \tau \ \textup{Re}(z)^2 }{1-\tau^2} \right]  \exp \left[ \frac{z^2 +\bar{z}^2}{2\tau} \right] \int_{-\infty}^{\infty} dp  \int_{-\infty}^{\infty} dq \ \Theta_N^{(1)}\left(\frac{p^2 - q^2}{2}\right) \nonumber \\
        &\times \exp \left[ \frac{N}{2} \left( p^2 - q^2 \right) -\frac{N}{2\tau }\left( p^2 + q^2 \right) - i \frac{\sqrt{2N}\RE(z)}{\tau} \, q  - \frac{\sqrt{2N}\IM(z)}{\tau} \, p  \right]  \ .
        \label{eq:rho_WNH_start}
    \end{align}
    Starting from here and using  Eq. \eqref{eq:WNH_tau_z} we see that asymptotically 
    \be\label{eq:WNH_eGinUE_prefacs}
        \frac{N}{2\pi^2\tau} \frac{1}{\sqrt{1-\tau^2}} \exp \left[ -\frac{ \vert z \vert^2 - \tau \ \textup{Re}(z)^2 }{1-\tau^2} \right] \exp[\frac{z^2 +\bar{z}^2}{2\tau}] \approx \frac{N\sqrt{N}}{2\pi^3 |\alpha|} \exp \left[ \frac{NX^2}{2} + \frac{3 \pi^2 \alpha^2 X^2}{8} - \frac{2y^2}{\alpha^2} \right] \ ,
    \ee
    where we have neglected terms $O(1/N)$ and smaller inside the exponential. We also observe that the exponential term inside the integral over $p$ has then a finite limit as $N\to \infty$:
    \be
      \lim_{N\to \infty}  \exp \left[- N \left( \frac{1-\tau}{2\tau }p^2 + \frac{\sqrt{2}\IM(z)}{\tau \sqrt{N}} \, p \right) \right] = \exp \left[- \frac{\pi^2 \alpha^2}{4}\, p^2 - \sqrt{2}\pi y \, p  \right] \ ,
    \ee
    whereas for the integral over $q$ we can use the leading-order asymptotic: 
    \be
       \exp \left[- N \left( \frac{1+\tau}{2\tau }q^2 + i \frac{\sqrt{2}\RE(z)}{\tau \sqrt{N}} \, q \right)  \right] \approx \exp \left[ - \left( N q^2 + i \sqrt{2} N X q  + \frac{\pi^2\alpha^2}{4} q^2 + i \frac{\pi^2 \alpha^2 X q}{ \sqrt{2}} \right) \right] \ .
    \ee
     We then see that the integral over $q$ can be immediately evaluated via Laplace's method around the point $q^\star = -iX/\sqrt{2}$. Similar forms to this integral are needed later, hence we write
    \begin{equation} 
    \begin{split}
        \int_{-\infty}^{\infty} dq f(p,q) \, \Theta_N^{(M)}\left( \frac{1}{2}\left( p^2 -q^2 \right) \right) &\exp \left[- N \left( \frac{1+\tau}{2\tau }q^2 + i \frac{\sqrt{2}\RE(z)}{\tau \sqrt{N}} \, q \right)  \right] \approx \\
        &\sqrt{\frac{\pi}{N}}f\left( p, -\frac{iX}{\sqrt{2}} \right) \Theta_N^{(M)}\left( \frac{p^2}{2} + \frac{X^2}{4} \right) \exp[ - \frac{NX^2}{2} - \frac{3\pi^2 \alpha^2 X^2}{8}] \ , 
    \end{split} \label{eq:int_q_saddle}
    \end{equation}
    where for the purpose of evaluating Eq. \eqref{eq:rho_WNH_start} we use $f(p,q) = 1$ and $M=1$. Putting the previous steps together we find that for large $N$, asymptotically
    \be
        \rho^{(\textup{eGinUE},c)}_N(z) \approx \frac{N}{\pi^2} \frac{1}{2\sqrt{\pi}|\alpha|} \, e^{- \frac{2y^2}{\alpha^2}} \int_{-\infty}^{\infty} dp \, \Theta_N^{(1)}\left( \frac{p^2}{2} + \frac{X^2}{4}  \right) \exp \left[-\frac{\pi^2 \alpha^2}{4}\, p^2 - \sqrt{2}\pi y \, p \right] \ .
        \label{eq:rho_WNH_approx}
    \ee
    As $N\to \infty$ the function $\Theta_N^{(M)}(s)$ approaches the $N-$independent Heaviside step function, $\Theta[1-s]$, effectively restricting the $p$ integration to the interval $[-\sqrt{2(1-X^2/4)},\sqrt{2(1-X^2/4)}]$. Using the symmetries of the integrand we finally see that
    \be\label{Eq:eGinUEbulkweakproof}
      \lim_{N\to \infty}  \frac{\pi^2}{N} \, \rho^{(\textup{eGinUE},c)}_N(z)  = \frac{\sqrt{2}}{\pi^{3/2}}  \frac{1}{|\alpha|} \, e^{- \frac{2y^2}{\alpha^2}} \, \int_{0}^{\pi\sqrt{1 - \frac{X^2}{4}}} du \exp \left[-\frac{\alpha^2}{2} u^2  \right] \cosh(2\, y \, u) \ , 
    \ee
    thus proving Eq. \eqref{Eq:eGinUEdensityweakbulk}. In order to evaluate the mean self-overlap at WNH one must now evaluate $R_N$ in this limit and to do so, we begin by using the integral representation presented in Lemma 4.4. The evaluation of these integrals follows much the same procedure as what was done for $\rho_{N}^{\text{(eGinUE,c)}}(z)$, the difference being that we now use $f(p,q) = (p^2-q^2)/2$. After appropriate adjustments, this yields the asymptotic 
    \begin{equation}
        R_N \approx  N \frac{N}{\pi^2} \frac{\sqrt{2}}{\pi^{3/2}}  \frac{1}{|\alpha|}  \, e^{- \frac{2y^2}{\alpha^2}}  \int_{0}^{\pi \sqrt{1 - \frac{X^2}{4}}} du \, \left( \frac{u^2}{\pi^2} + \frac{X^2}{4}  \right) \exp \left[-\frac{\alpha^2 u^2}{2} \right] \cosh(2yu) \ .
    \end{equation}
    Substituting this and Eq. \eqref{eq:rho_WNH_approx} into Eq. \eqref{Eq:MainResCOMPLEX} and collecting the leading-order terms 
    finally gives 
    \begin{equation}
        \begin{split}
        \lim_{N\to \infty}    \frac{\pi^2}{N} \mathcal{O}^{(\textup{eGinUE},c)}_N(z) \rightarrow \frac{\sqrt{2}}{\pi^{3/2}}  \frac{1}{|\alpha|}  \, e^{- \frac{2y^2}{\alpha^2}}  \int_{0}^{\pi \sqrt{1 - \frac{X^2}{4}}} du \, e^{-\frac{\alpha^2 u^2}{2} } \cosh(2yu)  \left[ 1 + \pi^2 \alpha^2\left( 1 - \frac{u^2}{\pi^2} - \frac{X^2}{4}  \right) \right] \ , 
        \end{split}
    \end{equation}
    as required.

\end{proof}


\subsubsection{eGinOE: Proof of Corollary \ref{cor:eGinOEbulkWeak}}\label{subsubsec:ProofeGinOEbulkWeak}

\begin{proof}
    In order to evaluate the mean self-overlap of eigenvectors in the eGinOE at WNH, defined by the parameters in Eq. \eqref{eq:WNH_tau_z}, we inspect the finite $N$ expression for the mean self-overlap, given in Theorem \ref{thm:MainRes}. The analysis then requires extracting the leading asymptotic behaviour of the prefactors and the quantities $P_N$ and $T_N$ in the limit of WNH. We begin by evaluating the prefactors in this limit and find that
    \be
    \begin{split}
        \frac{1}{\pi} \ \sqrt{\frac{1-\tau}{1+\tau}}& \ \exp \left[ -\frac{\RE(z)^2}{1+\tau} -\frac{\IM(z)^2}{1-\tau} \right] \bigg[ 1 + \sqrt{\frac{\pi(1-\tau^2)}{2}} \  \frac{e^{ \frac{2 \IM(z)^2}{1-\tau^2}} }{2\vert \IM(z) \vert} \ \textup{erfc}\left( \sqrt{\frac{2 }{1-\tau^2}} \ \vert \IM(z) \vert \right) \bigg] \\
        & \approx  \frac{|\alpha|}{2\sqrt{N}} \exp \left[ - \frac{NX^2}{2} - \frac{\pi^2 \alpha^2 X^2}{8}- \frac{2y^2}{\alpha^2} \right] \left[  1 + \sqrt{\frac{\pi}{2}} \, \frac{|\alpha|}{2 \vert y \vert} \, \exp \left[ \frac{2 y^2}{\alpha^2} \right] \,  \textup{erfc}\left(  \frac{\sqrt{2} \vert y \vert}{|\alpha|}  \right) \right]  \ .
    \end{split}
    \ee
    Now we proceed with the asymptotic evaluation of $P_N$ using the integral representation given in Lemma \ref{lem:PNintegralrep}. Inserting the parameterisation of $\tau$ and $z$ at WNH, the integral over $q$ can be evaluated using Eq. \eqref{eq:int_q_saddle} with $M=0$ and $f(p,q) = 1$, yielding after due manipulations
    \begin{equation}
    \begin{split}
        P_N \approx \frac{N}{\pi^2}  \frac{\sqrt{2N}}{\pi^{3/2} y} \exp[ \frac{NX^2}{2} + \frac{\pi^2 \alpha^2 X^2}{8}] \int_{0}^{\pi \sqrt{1 - \frac{X^2}{4}}} du \ u \exp[- \frac{\alpha^2 u^2}{2}] \sinh( 2 y u ) \ .
    \end{split} \label{eq:P_WNH}
    \end{equation}
    We are then left with the task of investigating $T_N$ in the limit of WNH, to do so we begin by using Lemma \ref{lem:TNintegralrep} for the necessary integral representation. Using the same approach as previously one can see that the integral over $q$ can be calculated by using $M=1$ and $f(p,q) = (p^2 - q^2)/2$ in Eq. \eqref{eq:int_q_saddle}, eventually this yields
    \begin{equation}
        T_N \approx \frac{N}{\pi^2} \frac{N \sqrt{2N}}{\pi^{3/2} y} e^{\frac{NX^2}{2} + \frac{\pi^2 \alpha^2 X^2}{8} } \int_{0}^{\pi \sqrt{1 - \frac{X^2}{4}}} du \ u \left( \frac{u^2}{\pi^2} + \frac{X^2}{4} \right) \exp[ - \frac{\alpha^2 u^2}{2}] \sinh(2 y u)   \ .
    \end{equation}
    Therefore, using the previous three results, one can see that    
    \be\label{Eq:MainResWeakproof2}
    \begin{split}
        \lim_{N\to \infty}\frac{\pi^2}{N}\mathcal{O}^{(\eGinOE,c)}_{N}(z) 
        &\rightarrow  \frac{1}{\sqrt{2}\pi^{3/2}} \ \frac{|\alpha|}{y} e^{- \frac{2y^2}{\alpha^2} } \left[ 1 + \sqrt{\frac{\pi}{2}} \, \frac{|\alpha|}{2\vert y \vert} \, \exp \left[ \frac{2 y^2}{\alpha^2} \right] \,  \textup{erfc}\left(  \frac{\sqrt{2}\vert y \vert}{|\alpha|} \right) \right] \\
        &\times \ \int_{0}^{\pi\sqrt{1 - \frac{X^2}{4}}}du \, u \, e^{-\frac{\alpha^2}{2}u^2} \, \sinh{\left( 2  y u \right)} \bigg[ 1 + \pi^2 \alpha^2  \left( 1 - \frac{X^2}{4} - \frac{u^2}{\pi^2}  \right)   \bigg] \ ,
    \end{split}
    \ee
    in accordance with  Eq. \eqref{Eq:eGinOEbulkWeakRes}.
\end{proof}

\noindent
Note that as a useful by-product of the proof, we can easily derive the known density of complex eGinOE eigenvalues at WNH, Eq. \eqref{Eq:eGinOEdensityweakbulk}. This can be done by utilising the asymptotic behaviour of $P_N$ at WNH, given in Eq. \eqref{eq:P_WNH}, in the finite $N$ expression for the density, Eq. \eqref{eGinOEdensityComplex}.


\subsubsection*{Acknowledgements}

We would like to thank Wojciech Tarnowski for useful discussions. This research has been supported by the EPSRC Grant EP/V002473/1 “Random Hessians and Jacobians: theory and applications”.





\begin{thebibliography}{99}
\bibitem{non-Herm_rev} Y. Ashida, Z. Gong, and M. Ueda, 
Non-Hermitian physics,
Adv. Phys. 69, 249--435 (2020).

\bibitem{KS}
B. A. Khoruzhenko and H.-J. Sommers,
Non-Hermitian ensembles, Chapter 18 in {\it The Oxford Handbook of Random Matrix Theory},
eds. G. Akemann, J. Baik and P. Di Francesco (Oxford University Press, Oxford, 2011).

\bibitem{May72}
R. M. May,
Will a large complex system be stable?,
{\it Nature} {\bf 238(5364)} (1972) 413-414.

\bibitem{SCSS}
H.-J. Sommers, A. Crisanti, H. Sompolinksy and Y. Stein,
Spectrum of Large Random Asymmetric Matrices,
{\it Phys. Rev. Lett.} {\bf 60} (1988) 1895.

\bibitem{WT13} G. Wainrib and J. Touboul,
Topological and Dynamical Complexity of Random Neural Networks,
{\it Phys. Rev. Lett.} {\bf 110}, 118101 (2013).

\bibitem{FK2016}
Y. V. Fyodorov and B. A. Khoruzhenko,
Nonlinear analogue of the May – Wigner instability transition,
{\it PNAS} {\bf 113(25)} (2016) 6827-6832.

\bibitem{MB2019}
J. Moran and J. P. Bouchaud,
May’s instability in large economies,
{\it Phys. Rev. E} {\bf 100(3)} (2019) 032307.

\bibitem{BAFK21} 
G. Ben Arous, Y. V. Fyodorov and B. A. Khoruzhenko, 
Counting equilibria of large complex systems by instability index,
{\it PNAS} {\bf 118} (34), e2023719118 (2021).

\bibitem{FSom1}
Y. V. Fyodorov and H.-J. Sommers,
Statistics of resonance poles, phaseshifts and time delays in quantum chaotic scattering: random matrix approach for systems with broken time-reversal invariance,
{\it J. Math. Phys.} {\bf 38} (1997) 1918–-1981.

\bibitem{Rotter}
I. Rotter,
A non-Hermitian Hamilton operator and the physics of open quantum systems,
{\it J. Phys. A Math. Theor.} {\bf 42} (2009)  153001.

\bibitem{FSav1}
Y. V. Fyodorov and D. V.  Savin,
Resonance scattering in chaotic systems, chapter 34 in: {\it The Oxford Handbook of Random Matrix Theory},
eds. G. Akemann, J. Baik and P. Di Francesco, (Oxford University Press, Oxford, 2011), p. 703.

\bibitem{random_Lindblad} 
S. Denisov, T. Laptyeva, W. Tarnowski, D. Chruściński, and K. Życzkowski,
Universal Spectra of Random Lindblad Operators,
{\it Phys. Rev. Lett.} {\bf 123}, 140403 (2019).

\bibitem{AKMP}
G. Akemann, M. Kieburg, A. Mielke and T. Prosen,
Universal signature from integrability to chaos in dissipative open quantum systems, 
{\it Phys. Rev. Lett.} {\bf 123(25)} (2019) 254101.

\bibitem{SRP}
L. Sa, P. Ribeiro and T. Prosen, 
Complex spacing ratios: A signature of dissipative quantum chaos, 
{\it Phys. Rev. X } {\bf 10(2)} (2020) 021019.

\bibitem{LPC21} J. Li, T. Prosen, and A. Chan,
Spectral Statistics of Non-Hermitian Matrices and Dissipative Quantum Chaos,
{\it Phys. Rev. Lett.} {\bf 127}, 170602 (2021).

\bibitem{CM}
J. T. Chalker and B. Mehlig,
Eigenvector Statistics in Non-Hermitian Random Matrix Ensembles,
{\it Phys. Rev. Lett.} {\bf 81}  (1998)  3367--3370.

\bibitem{MC}
B. Mehlig and J. T. Chalker,
Statistical properties of eigenvectors in non-Hermitian Gaussian random matrix ensembles,
{\it J. Math. Phys.} {\bf 41} (2000) 3233--3256.

\bibitem{Ginibre}
J. Ginibre,
Statistical ensembles of complex, quaternion and real matrices,
{\it J. Math. Phys.} {\bf  6} (1965) 440--449.

\bibitem{BD}
P. Bourgade and G. Dubach,
The distribution of overlaps between eigenvectors of Ginibre matrices,
{\it Probab. Theory Relat. Fields} {\bf 177} (2020) 397–-464.

\bibitem{FyodorovCMP}
Y. V. Fyodorov,
On statistics of bi-orthogonal eigenvectors in real and complex Ginibre ensembles: combining partial Schur decomposition with supersymmetry,
{\it Commun. Math. Phys.} {\bf 363} (2018) 579--603.

\bibitem{SS}
D. V. Savin and V. V. Sokolov,
Quantum versus classical decay laws in open chaotic systems,
{\it Phys. Rev. E} {\bf 56} (1997) R4911–-R4913.

\bibitem{JNNPZ}
R. A. Janik, W. N\"orenberg, M. A.  Nowak, G. Papp and I. Zahed,
Correlations of eigenvectors for non-Hermitian random matrix models,
{\it Phys. Rev. E} {\bf 60} (1999) 2699 -- 2705.

\bibitem{SFPB}
H. Schomerus, K. Frahm, M. Patra and C. W. J.  Beenakker,
Quantum limit of the laser line width in chaotic cavities and statistics of residues of scattering matrix poles,
{\it Physica A} {\bf A278} (2000) 469 --496.

\bibitem{PSB}
M. Patra, H. Schomerus and C. W. J. Beenakker,
Quantum-limited linewidth of a chaotic laser cavity,
{\it Phys. Rev. A} {\bf 61} (2000)  023810.

\bibitem{MS}
B. Mehlig and M. Santer,
Universal eigenvector statistics in a quantum scattering ensemble,
{\it Phys. Rev. E} {\bf 63} (2001)  020105(R).

\bibitem{FM}
Y. V. Fyodorov and B. Mehlig,
Statistics of resonances and nonorthogonal eigenfunctions in a model for single-channel chaotic scattering,
{\it Phys. Rev. E}  {\bf 66} (2002) 045202.

\bibitem{FSav2}
Y. V. Fyodorov and D. V. Savin,
Statistics of resonance width shifts as a signature of eigenfunction non-orthogonality,
{\it Phys. Rev. Lett. } {\bf 108} (2012) 184101.

\bibitem{GKLMRS}
J.-B. Gros, U. Kuhl, O. Legrand, F. Mortessagne, E. Richalot and D. V. Savin,
Experimental width shift distribution: a test of nonorthogonality for local and global perturbations,
{\it Phys. Rev. Lett. } {\bf 113} (2014) 224101.

\bibitem{BGNTW}
Z. Burda, J. Grela, M. A. Nowak, W. Tarnowski and P. Warchol,
Dysonian dynamics of the Ginibre  ensemble, 
{\it Phys. Rev. Lett. } {\bf 113} (2014)  104102.

\bibitem{BGNTW2}
Z. Burda, J. Grela, M. A. Nowak, W. Tarnowski, and P. Warchol,
Unveiling the significance of eigenvectors in diffusing non-Hermitian matrices by identifying the underlying Burgers dynamics,
{\it Nucl. Phys. B } {\bf 897} (2015) 421–-447.

\bibitem{WS}
M. Walters and S. Starr,
A note on mixed matrix moments for the complex Ginibre ensemble,
{\it J. Math. Phys.} {\bf 56} (2015) 013301.

\bibitem{BNST}
S. Belinschi, M. A. Nowak, R. Speicher and W. Tarnowski,
Squared eigenvalue condition numbers and eigenvector correlations from the single ring theorem,
{\it J. Phys. A: Math. Theor.} {\bf 50} (2017) 105204.

\bibitem{BSV}
Z. Burda, B. J.  Spisak and P. Vivo,
Eigenvector statistics of the product of Ginibre matrices,
{\it Phys. Rev. E}  {\bf 95(2)} (2017) 022134.

\bibitem{Grela}
J. Grela,
What drives transient behavior in complex systems?, 
{\it Phys. Rev. E } {\bf 96} (2017) 022316.

\bibitem{GW}
J. Grela and P. Warchol,
Full Dysonian dynamics of the complex Ginibre ensemble,
{\it J. Phys. A: Math. Theor.} {\bf 51} (2018) 425203.

\bibitem{NT}
M. A. Nowak and W. Tarnowski,
Probing non-orthogonality of eigenvectors in non-Hermitian matrix models: diagrammatic approach,
{\it JHEP} {\bf 2018} (2018) 152. 

\bibitem{BZ}
F. Benaych-Georges and O. Zeitouni,
Eigenvectors of non normal random matrices,
{\it Electron. Commun. Probab.}  {\bf 23} (2018) 1-12.

\bibitem{GOCNT}
E. Gudowska-Nowak, J. Ochab, D. Chialvo, M. A. Nowak and W. Tarnowski,
From synaptic interactions to collective dynamics in random neuronal networks models: critical role of eigenvectors and transient behavior, 
{\it Neural Comput.} {\bf 32} (2020) 395--423.

\bibitem{Yabuoku} 
S. Yabuoku,
Eigenvalue processes of Elliptic Ginibre Ensemble and their overlaps,
{\it Int. J. of Math. Ind. } {\bf 12} (2020) 2050003.

\bibitem{AFK}
G. Akemann, Y.-P. F\"orster and M. Kieburg,
Universal eigenvector correlations in quaternionic Ginibre ensembles,
{\it J. Phys. A: Math.Theor.} {\bf 53} (2020) 145201, 26 pp.

\bibitem{ATTZ}
G. Akemann, R. Tribe, A. Tsareas and O. Zaboronski,
On the determinantal structure of conditional overlaps for the complex Ginibre ensemble,
{\it Random Matrix Th. Appl.} {\bf 9} (2020) 2050015.

\bibitem{Dubach21a}
G. Dubach,
On eigenvector statistics in the spherical and truncated unitary ensembles,
{\it Electron. J. Probab.} {\bf 26} (2021) 1--29.

\bibitem{Dubach21b}
G. Dubach,
Symmetries of the quaternionic Ginibre ensemble,
{\it Random Matrices Theory Appl.} {\bf 10} (2021) 2150013.

\bibitem{FT}
Y. V. Fyodorov and W. Tarnowski,
Condition numbers for real eigenvalues in the real Elliptic Gaussian ensemble,
{\it Ann. Henri Poincar\'e} {\bf 22} (2021)  309--330.

\bibitem{CR}
N. Crawford and R. Rosenthal,
Eigenvector correlations in the complex Ginibre ensemble,
{\it Ann. Appl. Probab.} {\bf  32} (2022) 2706-2754.

\bibitem{CS2022}
G. Cipolloni, L. Erd\"os and D. Schr\"oder,
On the condition number of the shifted real Ginibre ensemble, 
{\it SIAM J. Matr. Anal. Applic.} {\bf 43} (2022) 1469–1487.

\bibitem{FyoOsm22}
Y. V. Fyodorov and M. Osman,
Eigenfunction non-orthogonality factors and the shape of CPA-like dips in a single-channel reflection from lossy chaotic cavities,
{\it J. Phys. A: Math. and Theor.} {\bf 55} (2022) 224013.

\bibitem{Dubach23} 
G. Dubach,
Explicit formulas concerning eigenvectors of weakly non-unitary matrices,
{\it Electron. Commun. Probab.} {\bf 28} (2023) 1--11.

\bibitem{Cipolloni23a}
G. Cipolloni and J. Kudler-Flam,
Entanglement Entropy of Non-Hermitian Eigenstates and the Ginibre Ensemble,
{\it Phys. Rev. Lett.} {\bf 130} (2023) 010401.

\bibitem{Cipolloni23b}
G. Cipolloni and J. Kudler-Flam,
Non-Hermitian Hamiltonians Violate the Eigenstate Thermalization Hypothesis,
{\it Phys. Rev. B} {\bf 109}, L020201 (2024)

\bibitem{CEHS2023}
C. Cipolloni, L. Erd\"os, J. Henheik and D. Schr\"oder,
Optimal Lower Bound on Eigenvector Overlaps for non-Hermitian Random Matrices, 
arXiv:2301.03549.

\bibitem{EJ2023}
L. Erd\"os and H.C. Ji,
Wegner estimate and upper bound on the eigenvalue condition number of non-Hermitian random matrices, arXiv:2301.04981.

\bibitem{Esaki}
L. Esaki, M. Katori and S. Yabuoku,
Eigenvalues, eigenvector-overlaps, and regularized Fuglede-Kadison determinant of the non-Hermitian matrix-valued Brownian motion, 
arXiv:2306.00300.

\bibitem{WCF}
T. R. W\"urfel, M. J. Crumpton and Y. V.Fyodorov, 
Mean left-right eigenvector self-overlap in the real Ginibre ensemble,
arXiv:2310.04307.

\bibitem{FGNNW}
Y. V. Fyodorov, E.  Gudowska-Nowak, M. A. Nowak and W. Tarnowski,
Fluctuation-dissipation relation for non-Hermitian Langevin dynamics,
arXiv:2310.09018.

\bibitem{Noda23a}
K. Noda,
Determinantal structure of the conditional expectation of the overlaps for the induced Ginibre unitary ensemble, arXiv:2310.15362.

\bibitem{Noda23b}
K. Noda,
Determinantal structure of the overlaps for induced spherical unitary ensemble,
arXiv:2312.12690.

\bibitem{GKR23} 
S. Ghosh, M. Kulkarni and S. Roy,
Eigenvector Correlations Across the Localisation Transition in non-Hermitian Power-Law Banded Random Matrices, arXiv:2304.09892.

\bibitem{Tarnowski24}
W. Tarnowski,
Condition numbers for real eigenvalues of real elliptic ensemble: weak non-normality at the edge, arXiv:2401.03249.

\bibitem{Girko1} 
V. L. Girko, 
Circular law,
{\it Theor. Probab. and Applic.} {\bf 29} (1985) 694--706.

\bibitem{Girko2} 
V. L. Girko, 
Elliptic law,
{\it Theor. Probab. and Applic.} {\bf 30} (1986) 640--651.

\bibitem{FKS97a}
Y. V. Fyodorov, B. A. Khoruzhenko and H.-J. Sommers,
Almost Hermitian random matrices: crossover from Wigner-Dyson to Ginibre eigenvalue statistics,
{\it Phys. Lett. A} {\bf 226}  (1-2), 46--52 (1997). 

\bibitem{FKS97b}
Y. V. Fyodorov, B. A. Khoruzhenko and H.-J. Sommers,
Almost Hermitian random matrices: crossover from Wigner-Dyson to Ginibre eigenvalue statistics,
{\it Phys. Rev. Lett.} {\bf 79} (1997)  557--560.

\bibitem{FKS98}
Y. V. Fyodorov, B. A. Khoruzhenko and H.-J. Sommers,
Universality in the random matrix spectra in the regime of weak non-Hermiticity,
{\it Ann. Inst. H. Poincaré Phys. Théor.} {\bf 68 } (1998) 449–489.

\bibitem{BF1}
S. S. Byun and P. J. Forrester,
Progress on the study of the Ginibre ensembles I: GinUE, 
arXiv:2211.16223.

\bibitem{BF2}
S. S. Byun and P. J. Forrester,
Progress on the study of the Ginibre ensembles II: GinOE and GinSE, 
arXiv:2301.05022.

\bibitem{FN2}
P. J. Forrester and T. Nagao,
Skew orthogonal polynomials and the partly symmetric real Ginibre ensemble,
{\it J. Phys. A} {\bf 2008} {\it 41}, 375003.

\bibitem{APS}
G. Akemann, M. J. Phillips and H.-J. Sommers,
Characteristic polynomials in real Ginibre ensembles,
{\it J. Phys. A: Math. Theor.} {\bf 42} (2009) 012001.

\bibitem{ForM}
P. J. Forrester and A. Mays,
A Method to Calculate Correlation Functions for $\beta = 1$ Random Matrices of Odd Size,
{\it J. Stat. Phys. } {\bf 134} (2009) 443462.

\bibitem{ACV}
G. Akemann,  M. Cikovic and M. Venker,
Universality at weak and strong non-Hermiticity beyond the elliptic Ginibre ensemble,
{\it Commun. Math. Phys.} {\bf 362} no. 3 (2018) 1111-1141.

\bibitem{Efe97}
K.B. Efetov,
Directed Quantum Chaos,
{\it Phys. Rev. Lett.} {\bf 79} (1997) 491.

\bibitem{LR}
S. Y. Lee and R. Riser,
Fine asymptotic behavior for eigenvalues of random normal matrices: Ellipse case,
{\it J. Math. Phys.} {\bf 57} (2016) 023302.

\bibitem{TaoVu}
T. Tao and V. Vu,
Random Matrices: Universality of Local Spectral Statistics of Non-Hermitian Matrices,
{\it Ann. Prob.} {\bf 43} No. 2 (2015) 782-874.

\bibitem{EKS}
A. Edelman, E. Kostlan and M. Shub,
How many eigenvalues of a random matrix are real?,
{\it J. Amer. Math. Soc.} {\bf 7} (1994) 247--267.

\bibitem{Edelman}
A. Edelman,
The probability that a random real Gaussian matrix has k real eigenvalues, related distributions, and the circular law,
{\it J. Multivar. Anal.} {\bf 60} (1997)  203–-232.

\bibitem{FK}
Y. V. Fyodorov and B. A. Khoruzhenko, 
On absolute moments of characteristic polynomials of a certain class of complex random matrices,
{\it Commun. Math. Phys.} {\bf 273(3)} (2007) 561–599.

\bibitem{CSS}
S. Caracciolo, A. D. Sokal and A. Sportiello,
Algebraic/combinatorial proofs of Cayley-type identities for derivatives of determinants and pfaffians,
{\it Adv. Appl. Math.} {\bf 50} (2013) 474-594, arXiv:1105.6270.

\bibitem{DVarjas}
D. Varjas,
Generalizations of the Pfaffian to non-antisymmetric matrices, 
arXiv:2209.02578.

\bibitem{Grad}
I. S. Gradshteyn and I.M. Ryzhik,
{\it Tables of Integers, Series and Products} {\bf 7th ed.} (Academic Press, New York, 2000).

\bibitem{NIST}
F. W. J. Olver, D. W. Lozier, R. F. Boisvert and C. W. Clark (eds.),
{\it NIST Handbook of Mathematical Functions} (Cambridge University Press, Cambridge, 2010).

\bibitem{LC}
J. S. Lomont and M. S. Cheema,
Properties of Pfaffians,
The Rocky Mountain Journal of Mathematics {\bf 15(2)} (1985) 493–512.


\end{thebibliography}
\end{document}